\documentclass[12pt]{article}
\usepackage{amsmath,amssymb,amsthm}
\usepackage{txfonts,bm,pifont}
\usepackage[dvipdfmx]{graphicx}
\usepackage[bookmarks=true,bookmarksnumbered=true,colorlinks=true]{hyperref}
\usepackage{float}
\AtBeginDvi{}
\theoremstyle{plain}
\newtheorem{thm}{Theorem}[section]
\newtheorem{prop}[thm]{Proposition}

\newtheorem{lemma}[thm]{Lemma}

\newtheorem{rem}[thm]{Remark}
\newcommand{\CO}{\mathrm{CO}}

\newcommand{\Sim}{\mathrm{Sim}}
\newcommand{\R}{\mathbb{R}}
\newcommand{\Z}{\mathbb{Z}}
\makeatletter
\@addtoreset{equation}{section} 

\makeatother
\begin{document}
\begin{center}
\begin{Large}
%Conformal
Isogonal
%Equiangular
Deformation of Discrete Plane Curves\\
and Discrete Burgers Hierarchy\\[5mm]
\end{Large}
\begin{normalsize}
Kenji {\sc Kajiwara}\\
Institute of Mathematics for Industry, Kyushu University\\
744 Motooka, Fukuoka 819-0395, Japan\\
e-mail: kaji@imi.kyushu-u.ac.jp\\[2mm]
Toshinobu {\sc Kuroda}\\
Uwajima South Secondary School\\
5-1 Bunkyocho, Uwajima, Ehime 798-0066, Japan\\[2mm]
Nozomu {\sc Matsuura}\\
Department of Applied Mathematics, Fukuoka University\\
Nanakuma 8-19-1, Fukuoka 814-0180, Japan\\
e-mail: nozomu@fukuoka-u.ac.jp
\end{normalsize}
\end{center}
\begin{abstract}
We study deformations of plane curves in the similarity geometry.  It is
known that continuous deformations of smooth curves are described by the
Burgers hierarchy.  In this paper, we formulate the discrete deformation
of discrete plane curves described by the discrete Burgers hierarchy as
isogonal deformations. We also construct explicit formulas for the curve
deformations by using the solution of linear diffusion
differential/difference equations.
\end{abstract}

\section{Introduction}
Integrable deformations of curves play crucial roles in the differential geometry of space/plane
curves \cite{Rogers-Schief:book}. Formulating the deformation of curves as the simultaneous system
of the Frenet-Serret formula for the Frenet frame of curves and its deformation equation, it
naturally gives rise to various integrable systems. This framework can be discretized so that it is
consistent with the theory of discrete integrable systems, which is sometimes referred to as the
discrete differential geometry \cite{Bobenko-Suris:book}. Various deformations of discrete curves
have been formulated in this context
\cite{Doliwa-Santini:PLA,Doliwa-Santini:JMP,Doliwa-Santini:dsG,Fujioka-Kurose:Burgers,Hisakado-Nakayama-Wadati,Hisakado-Wadati,Hoffmann:dNLS,Hoffmann-Kutz,Nakayama:JPSJ2007,Nakayama_Segur_Wadati:PRL,Nishinari,Pinkall:dNLS}.
The theory of discrete differential geometry of curves is now making progress in explicit
constructions of curves, by using the theory of $\tau$ functions
%\cite{IKMO:KJM,IKMO:JPA,Matsuura:IMRN}.
\cite{IKMO:KJM,IKMO:JPA,IKMO:dmKdV_space_curve,Matsuura:IMRN}.

When we change the geometric structure of space/plane in the framework of Klein geometry, the curve
motions are governed by various integrable equations
\cite{Chou-Qu:2002_PD,Chou-Qu:2003_JNS,Chou-Qu:2004_CSF}. Therefore it may be an interesting and
important problem to discretize such deformations of curves consistently with corresponding
integrable structures.

In this paper, we consider deformation of the plane curves in the similarity geometry, which is a
Klein geometry associated with the linear conformal group. In this setting, it is known that the
Burgers hierarchy describes the deformations of similarity curvature of curves. We present discrete
deformations of discrete plane curves in the similarity geometry described by the discrete Burgers
hierarchy as the isogonal deformations in which each angle of adjascent segments is preserved.  The
lattice intervals of the hierarchy are generalized to arbitrary functions of corresponding
independent variables. Using this formulation, we present explicit formulas of curves for both
smooth and discrete cases. We note that the (complex) Burgers equation and its discrete analogue
also arise in the curve deformations in complex hyperbola, where the Hamiltonian formulation of the
deformation of smooth curves is discussed \cite{Fujioka-Kurose:Burgers}.

In Section \ref{sec:continuous}, we give a brief summary of deformation of smooth plane curves in
the similarity geometry, and we see that the Burgers hierarchy naturally arises as the equations for
the similarity curvature. We also construct the explicit formula for the family of plane curves
corresponding to the shock wave solutions to the Burgers equation. In Section
\ref{sec:discrete_curve}, we discretize the whole theory described in Section \ref{sec:continuous}
so that the deformations are governed by the discrete Burgers hierarchy.  Formulations of the Burgers
and the discrete Burgers hierarchies are discussed in detail in Appendix.

In \cite{Inoguchi:MEIS2015, Shimizu-Sato:JSIAM}, the deformation theory
of plane curves in the similarity geometry can be applied to the
construction and generalization of aesthetic curves in CAD. Also, in
\cite{FIKMO:hodograph,FIKMO:Dym,FMO:SP,FMO:CH,FMO:short_CH,FMO:SP_Lax}
discretizations for the class of nonlinear differential equations
describing the motions of plane curves are constructed by using the
geometric formulations, resulting in self-adaptive moving mesh discrete
model of the original equation. This discretization enables to contruct
highly accurate numerical scheme of given equation. The Burgers equation
is widely used as the universal model describing one-dimensional
nonlinear dissipative system after various transformations which are
difficult to discretize.  It may be possible to construct various useful
discrete models by using the result in this paper.  We hope that the
results in this paper serves as the basis of such industry-based
problems.

\section{Deformation of smooth curves}\label{sec:continuous}

Let $\gamma = \gamma (s)$ be a smooth curve in $\R^2$, $s$ be the
arc-length, and $\kappa$ be the curvature of $\gamma$.  We denote by
$\Sim (2)$ the similarity transformation group of $\R^2$, that is, $\Sim
(2) = \mathrm{CO} (2) \ltimes \R^2$ where $\mathrm{CO} (2)$ is the
linear conformal group
\begin{equation*}
\mathrm{CO} (2)
=
\left\{A \in \mathrm{GL} \left(2, \R\right)
\;\big|\;
{}^\mathrm{t}\!A A = c^2\, \mathrm{id} \, \text{ for some constant } c
\right\}.
\end{equation*}
The $\Sim (2)$-invariant parameter $x$ is given by the angle function
\begin{equation}\label{sim-param}
x = \int^s \kappa (s)\, ds,
\end{equation}
and the $\Sim (2)$-invariant curvature $u$ is defined as
\begin{equation}\label{sim-curv}
u =
\dfrac{1}{\kappa^2}
\dfrac{d \kappa}{ds}.
\end{equation}
The $x$ and $u$ are called the {\em similarity arc-length parameter} and
the {\em similarity curvature}, respectively.  If the similarity
curvature is constant $u = k_1$, then the inverse of Euclidean curvature
is $1/\kappa = - k_1 s + k_2$ for some constant $k_2$.  Thus $\gamma$ is
a log-spiral (if $k_1 \neq 0$) or a circle (if $k_1 = 0, k_2 \neq 0$).

The $\Sim (2)$-invariant frame $\phi = \left[T, N\right]$ is given by
\begin{equation}\label{eqn:T_N_smooth}
T = \gamma',\quad
N =
\begin{bmatrix}
0 & - 1\\
1 & 0
\end{bmatrix}
T,
\end{equation}
where the prime means differentiation with respect to the similarity
arc-length parameter $x$.  The $\mathrm{SO} (2)$-invariant frame (the
Frenet frame) $\phi_{\rm E}$ given by
\begin{equation*}
\phi_\mathrm{E}
= \left[T_\mathrm{E}, N_\mathrm{E}\right]
= \kappa \phi,\quad 
T_\mathrm{E} = \frac{d}{ds} \gamma,\quad N_\mathrm{E}=\begin{bmatrix}
0 & - 1\\
1 & 0
\end{bmatrix}T_\mathrm{E},
\end{equation*}
varies according to the Frenet formula
\begin{equation*}
\frac{d}{ds} \phi_\mathrm{E}
=
\phi_\mathrm{E}
\begin{bmatrix}
0 & - \kappa\\
\kappa & 0
\end{bmatrix}.
\end{equation*}
Therefore, by using \eqref{sim-param} and \eqref{sim-curv}, we have
\begin{equation}\label{lax-x}
\phi' = \phi
\begin{bmatrix}
- u & - 1\\
1 & - u
\end{bmatrix}.
\end{equation}

We denote by $\gamma (x, t)$ a deformation of a curve $\gamma (x)$.  We
use the dot to indicate differentiation with respect to time $t$.
Writing $\dot\gamma$ as the linear combination of $T$ and $N$ as
\begin{equation*}
\dot{\gamma}
=
f(x,t) T + g(x,t) N,
\end{equation*}
we have by using \eqref{eqn:T_N_smooth} that
\begin{equation}\label{lax-t}
\dot{\phi}
=
\phi
\begin{bmatrix}
f' - f u - g &
- g' + g u - f\\
g' - g u + f &
f' - f u - g
\end{bmatrix}.
\end{equation}
The compatibility condition of the linear system \eqref{lax-x} and
\eqref{lax-t} is given by
\begin{gather}
g' - g u + f
=\label{compat-1}
 a,\\
\dot{u} + \left(f' - f u - g\right)'
=\label{compat-2}
0,
\end{gather}
for some function $a = a (t)$.
Especially, choosing $f =  a - u,\, g = - 1$ and denoting $t=t_2$,
we have
\begin{align}
\frac{\partial \phi}{\partial t_2}
&=\label{eqn:phi_M_Burgers}
\phi
\begin{bmatrix}
- u' + u^2 + 1 - a u & -a\\
 a & - u' + u^2 + 1 - a u
\end{bmatrix},\\
\frac{\partial u}{\partial t_2}
&=\label{eqn:Burgers_with_a}
u'' - 2 u u' + a u'.
\end{align}
Equation \eqref{eqn:Burgers_with_a} is called the
\textit{Burgers equation},
which is linearized to
\begin{equation}
\frac{\partial}{\partial t_2} q
=\label{linear:Burgers}
\left(\frac{\partial^2}{\partial x^2} + 1
+ a \frac{\partial}{\partial x}\right) q,
\end{equation}
via the Cole-Hopf transformation \cite{Hopf}%\cite{MR0047234}
\begin{equation}\label{Cole-Hopf}
u
=
%- \dfrac{\partial}{\partial x} \log q.
- \left(\log q\right)'.
\end{equation}
Further, the Burgers hierarchy naturally arises as follows
\cite{Chou-Qu:2002_PD,Chou-Qu:2003_JNS,Chou-Qu:2004_CSF}. %\cite{MR1882237}.  
Substituting \eqref{compat-1} into \eqref{compat-2}, we have that
\begin{equation}\label{eqn:udot}
\dot{u}
=
\left(\Omega^2 + 1\right) g' + a u',
\end{equation}
where $\Omega = \partial_x - u - u' \partial_x^{- 1}$ is the recursion
operator of the Burgers hierarchy (see Appendix \ref{app:hier}). Here,
$\partial_x^{-1}$ is the formal integration operator with respect to
$x$, and in the following, the integration constant should be chosen to
be $0$. In view of this, we introduce an infinite number of time
variables $t=(t_2,t_3,t_4,\ldots)$, and choose $g' = \Omega^{i - 3} u'$
($i \geq 3$).  Then the higher flow with respect to the new time
variable $t_i$ is given by
\begin{equation}\label{def:higher-motion}
\frac{\partial}{\partial t_i} \phi
=
\phi
\begin{bmatrix}
- \partial_x^{- 1}
\left(\Omega^{i - 1} + \Omega^{i - 3} + a\right) u' &
-a\\
 a &
- \partial_x^{- 1}
\left(\Omega^{i - 1} + \Omega^{i - 3} + a\right) u'
\end{bmatrix}.
\end{equation}
The compatibility condition between \eqref{lax-x} and
\eqref{def:higher-motion} is the $i$-th Burgers equation
\begin{equation}\label{eqn:n-th_Burgers}
\frac{\partial}{\partial t_i} u
= \left(\Omega^{i - 1} + \Omega^{i - 3} + a\right) u',
\end{equation}
which is linearized to
\begin{equation}
\frac{\partial}{\partial t_i} q
=\label{linear}
\left(\frac{\partial^i}{\partial x^i}
+ \frac{\partial^{i - 2}}{\partial x^{i - 2}}
+ a \frac{\partial}{\partial x}\right) q,
\end{equation}
via the Cole-Hopf transformation \eqref{Cole-Hopf}. Note that the case
of $i=2$ of \eqref{linear} recovers \eqref{linear:Burgers}.

It is possible to express the postion vector $\gamma$ in terms of $q$ as
follows.  The inverse of Euclidean curvature satisfies $1/\kappa = c q$
for some function $c = c \left(t\right)$, because the similarity
curvature is logarithmic differentiation of $\kappa$, that is, $u$
satisfies that
\begin{equation*}
u = \dfrac{1}{\kappa^2} \dfrac{\partial \kappa}{\partial s}
= \dfrac{\kappa'}{\kappa^2} \dfrac{\partial x}{\partial s}
= \dfrac{\kappa'}{\kappa}
= \left(\log \kappa\right)'.
\end{equation*}
Since the similarity arclength parameter $x$ is the angle function,
we have
\begin{equation*}
T_\mathrm{E} =
\begin{bmatrix}
\cos(x+x_0)\\
\sin(x+x_0)
\end{bmatrix},
\end{equation*} 
where we have incorporated the ambiguity of
the angle function $x_0=x_0(t)$ explicitly. Hence
\begin{equation}\label{eqn:explict_formula_smooth}
\gamma
= \int^x T dx
= \int^x \dfrac{\,1\,}{\kappa}\, T_\mathrm{E}\, dx
= \int^x c(t) q(x,t)
\begin{bmatrix}
\cos(x+x_0(t))\\
\sin(x+x_0(t))
\end{bmatrix}
dx.
\end{equation}
We determine $c$ and $x_0$ by the deformation equation
\eqref{def:higher-motion}.  By differentiating $T$ by $t_i$ (here
$\dot{\phantom{c}}$ denotes $\partial_{t_i}$), we have by substituting
\eqref{eqn:explict_formula_smooth} into \eqref{def:higher-motion},
\begin{equation*}%\label{Tdot}
\dot{T} = (\dot{c} q + c\dot{q})
\begin{bmatrix}
\cos(x+x_0) \\ \sin(x+x_0)
\end{bmatrix}
+ cq\dot{x_0}
\begin{bmatrix}
-\sin(x+x_0) \\ \cos(x+x_0)
\end{bmatrix}
= \left(\frac{\dot{c}}{c}+\frac{\dot{q}}{q}\right) T
+ \dot{x_0} N.
\end{equation*}
%We compare \eqref{Tdot} with \eqref{def:higher-motion}. 
Note that from \eqref{eqn:n-th_Burgers} and \eqref{Cole-Hopf} we have $-
\partial_x^{- 1}\left(\Omega^{i - 1} + \Omega^{i - 3} + a\right)
u' = \dot{q}/q$.  Similarly for the case of $i=2$ we also have
$- u' + u^2 + 1 - a u = \dot{q}/q$ from \eqref{linear:Burgers}
and \eqref{Cole-Hopf}.  Then from \eqref{def:higher-motion} we obtain
\begin{equation*}
\dot{T}=\frac{\dot{q}}{q}T + aN,
\end{equation*}
which implies $c(t)=c$(const.) and $x_0 = A(t)$ where $\dot{A}(t)=a(t)$. Therefore we obtain:
%%%%%%%%%%%%%%%%%%%%%%%%%%%%%%%%%%%%
%
%%%%%%%%%%%%%%%%%%%%%%%%%%%%%%%%%%%%
\begin{prop} \label{prop:rep_formula:smooth}
Let $\gamma=\gamma(x,t)$, $t=(t_2,t_3,t_4,\ldots)$ be a position vector
of the plane curve in the similarity geometry satisfying \eqref{lax-x},
\eqref{eqn:phi_M_Burgers} and \eqref{def:higher-motion}. Then $\gamma$
admits the representation formula
\begin{equation}\label{eqn:representation}
 \gamma = \int^x c q(x,t)
\begin{bmatrix}
\cos\theta(x,t)\\
\sin\theta(x,t)
\end{bmatrix}
dx,\quad \theta(x,t)=x+A(t),
\end{equation}
where
\begin{equation*}
 \frac{\partial A(t)}{\partial t_i}=a(t), \quad i=2,3,4,\ldots,
\end{equation*}
$c$ is a constant, and $q(x,t)$ satisfies \eqref{linear}.
\end{prop}

For a shock wave solution to the Burgers hierarchy, we can explictly
construct the position vector.  For a positive integer $M$,
\begin{equation*}
q (x, t)
=
e^{t_2} + \sum_{k = 1}^M
\exp \left(\lambda_k x +
\sum_{i = 2}^\infty
\left({\lambda_k}^i + {\lambda_k}^{i - 2} + a\lambda_k\right) t_i
+ \xi_k\right),
\end{equation*}
solves the linear equation \eqref{linear}, where $\lambda_1, \xi_1,
\ldots, \lambda_M, \xi_M$ are parameters. Then
\eqref{eqn:representation} gives
\begin{align}
\gamma (x, t)
&=\nonumber
\int^x
c\, q (x, t)
\begin{bmatrix}
\cos\theta(x,t)\\
\sin\theta(x,t)
\end{bmatrix}
dx\\
&=\label{explicit}
c\,
\sum_{k = 0}^M
\dfrac{\exp \left(\lambda_k x
+ \sum_{i=2}^\infty (\lambda_k^i+\lambda_k^{i-2}+a) t_i + \xi_k\right)}%
{1 + {\lambda_k}^2}%\notag\\
%&\times
\begin{bmatrix}
\lambda_k \cos \theta + \sin\theta\\
\lambda_k \sin\theta - \cos \theta
\end{bmatrix},
\end{align}
where $\lambda_0 = \xi_0 = 0$.  Figure \ref{fig:Burgers_curve-1}, \ref{fig:Burgers_curve-2}
illustrate motion of plane curves corresponding to $M$-shock wave solutions ($M = 1, 2$,
respectively) of the Burgers equation \eqref{eqn:Burgers_with_a} with $t_i = 0 \left(i \geq
3\right)$.
%%%%%%%%%%%%%%%%%%%%%%%%%%%%%%%%%%%%%%%%%%
% Remark on a
%%%%%%%%%%%%%%%%%%%%%%%%%%%%%%%%%%%%%%%%%%
\begin{rem}\label{rem:a}
The parameter $a$ originally arises as an integration constant in \eqref{compat-1}, and play a role
of rotation in the deformation of smooth curves as seen in Proposition
\ref{prop:rep_formula:smooth}.  This parameter can be formally absorbed by a suitable linear
transformation of independent variables (see, for example, \eqref{eqn:n-th_Burgers} and
\eqref{linear}). In the discrete case, however, such manipulation is not applicable since the chain
rule does not work effectively. Actually the similar parameter appears in a non-trivial manner in
the deformation of discrete curves as shown in Section \ref{sec:discrete_curve}.
\end{rem}
%%%%%%%%%%%%%%%%%%%%%%%%%%%%%%%%%%%%%%%%%%
\begin{center}
\begin{figure}[h]
\begin{minipage}{.3\linewidth}
\includegraphics[width=\linewidth]{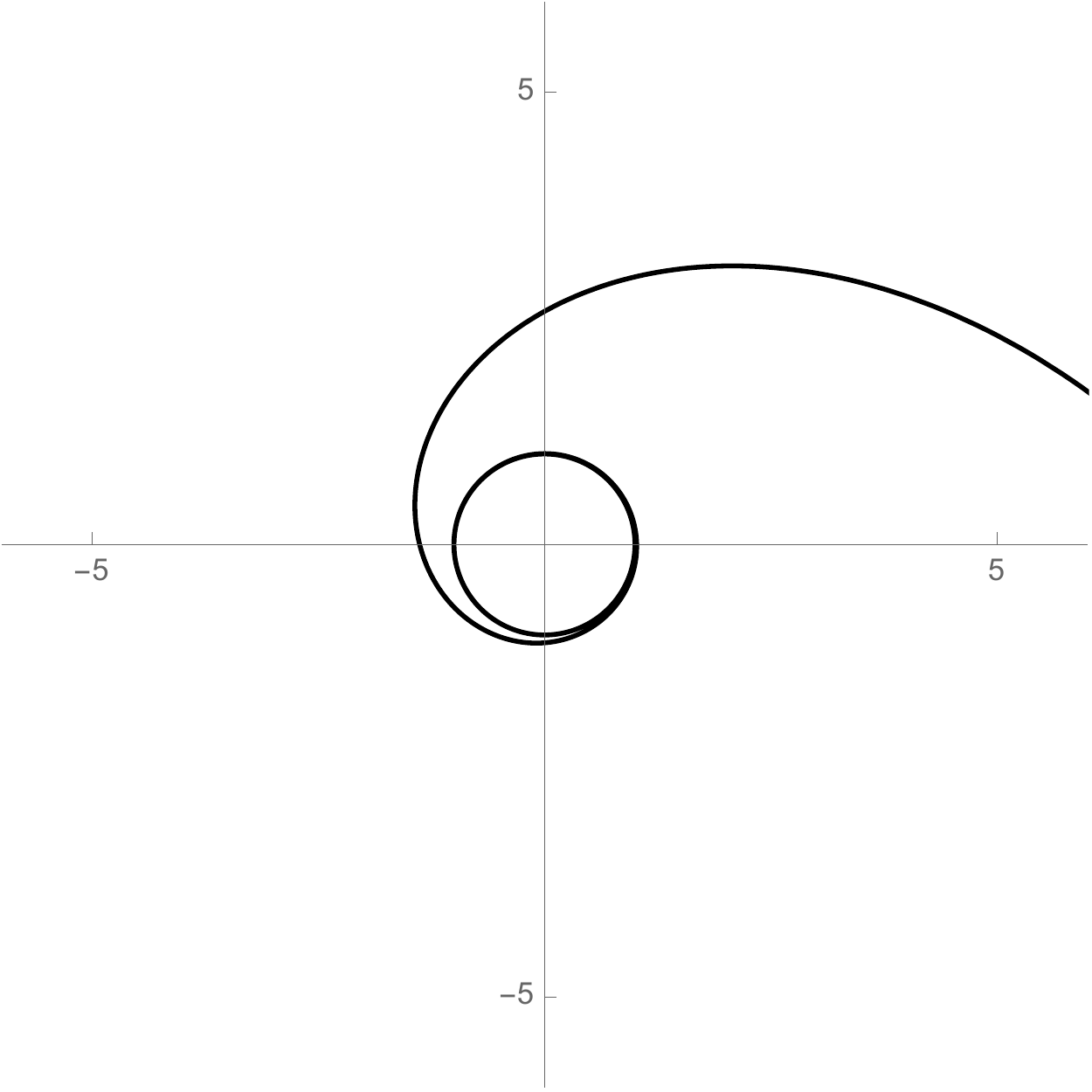}
\end{minipage}
\begin{minipage}{.3\linewidth}
\includegraphics[width=\linewidth]{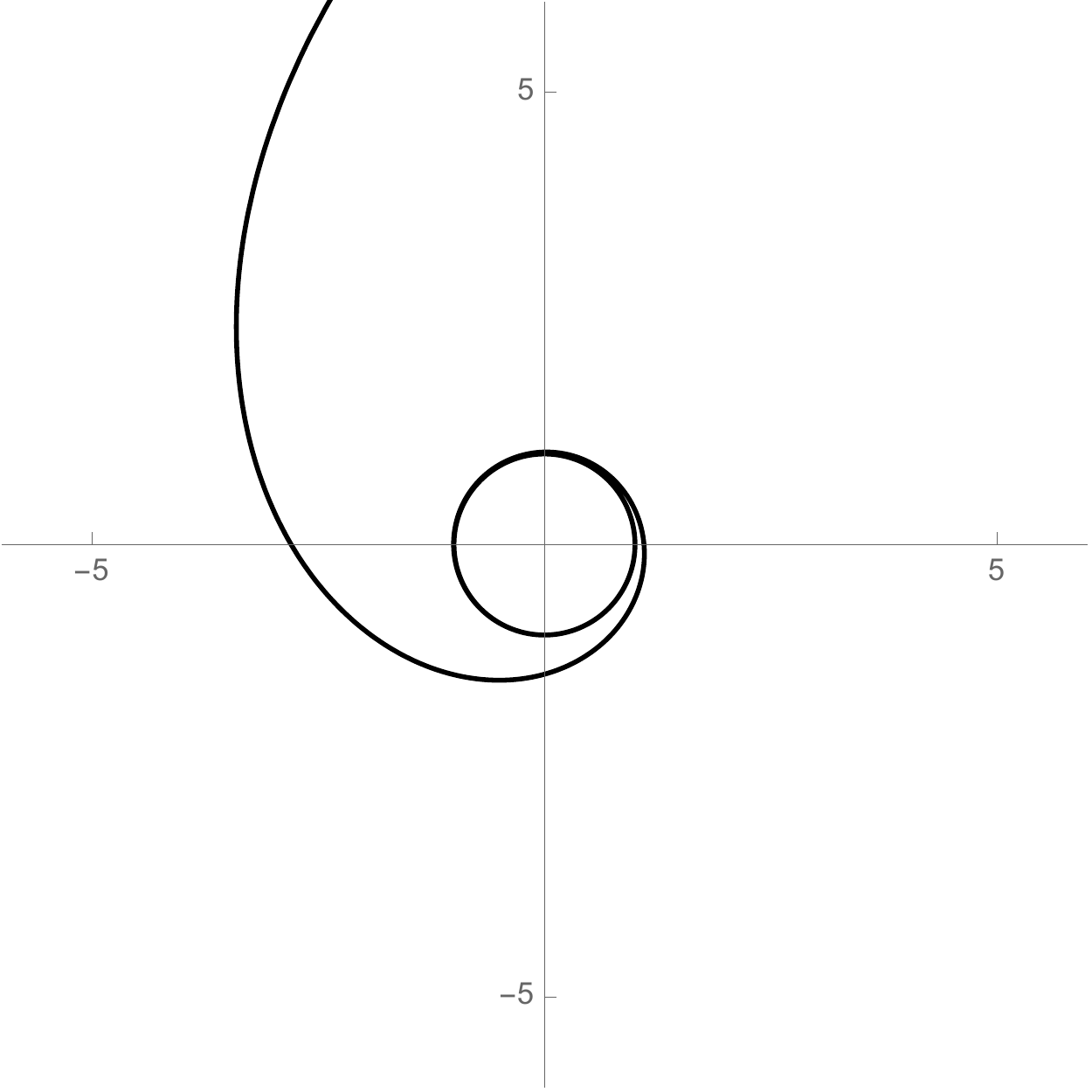}
\end{minipage}
\begin{minipage}{.3\linewidth}
\includegraphics[width=\linewidth]{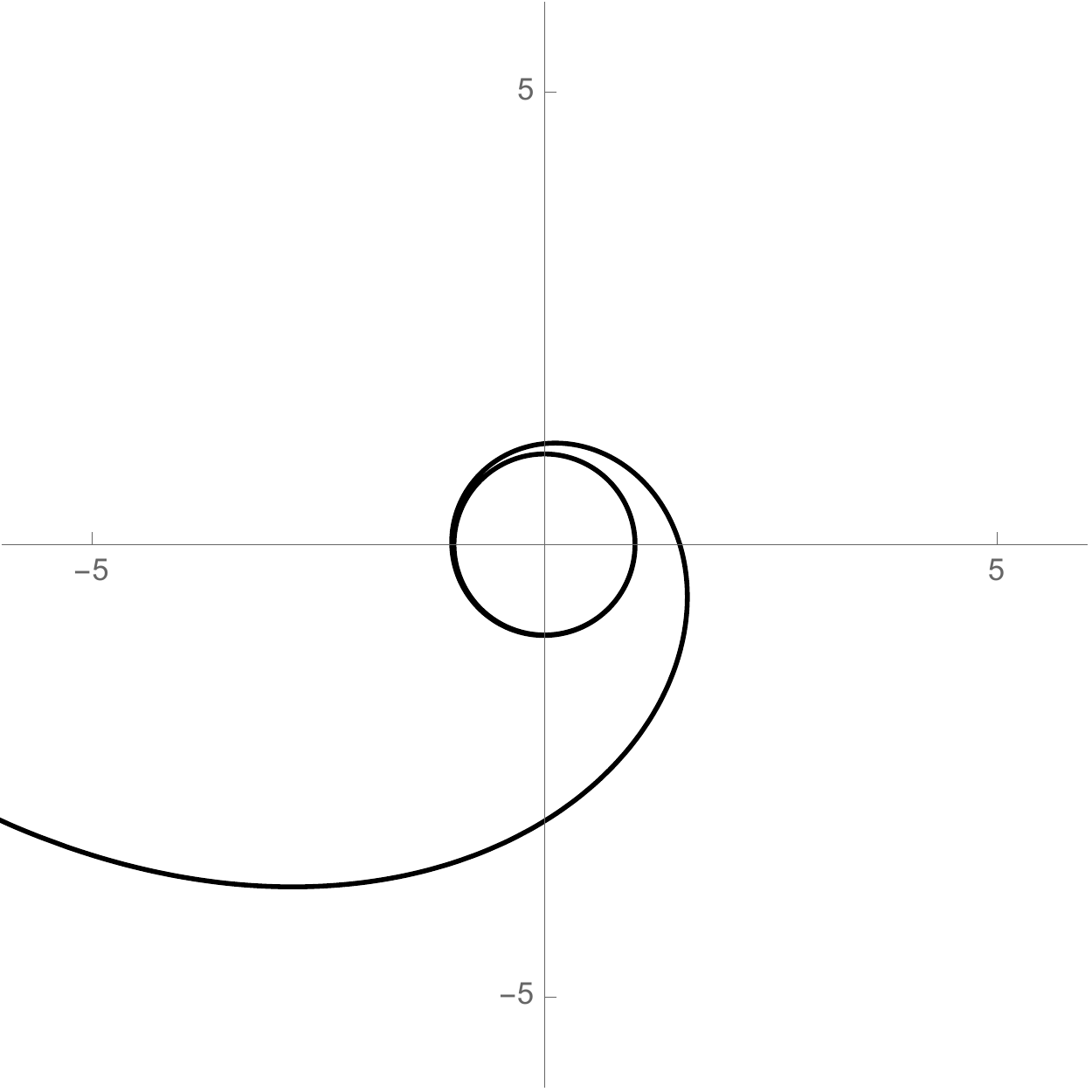}
\end{minipage}
\caption{Motion of plane curves $e^{- t_2} \gamma \left(x, t\right)$
corresponding to a 1-shock wave solution of the Burgers equation
\eqref{eqn:Burgers_with_a}. Parameters are $c = 1$,
$a = 0$,
$\lambda_1 = - 1$, $\xi_1 = 0$ and
$t_2=-8$ (left), $0$ (middle), $8$ (right).}
\label{fig:Burgers_curve-1}
\end{figure}
\end{center}
\begin{center}
\begin{figure}[h]
\begin{minipage}{.3\linewidth}
\includegraphics[width=\linewidth]{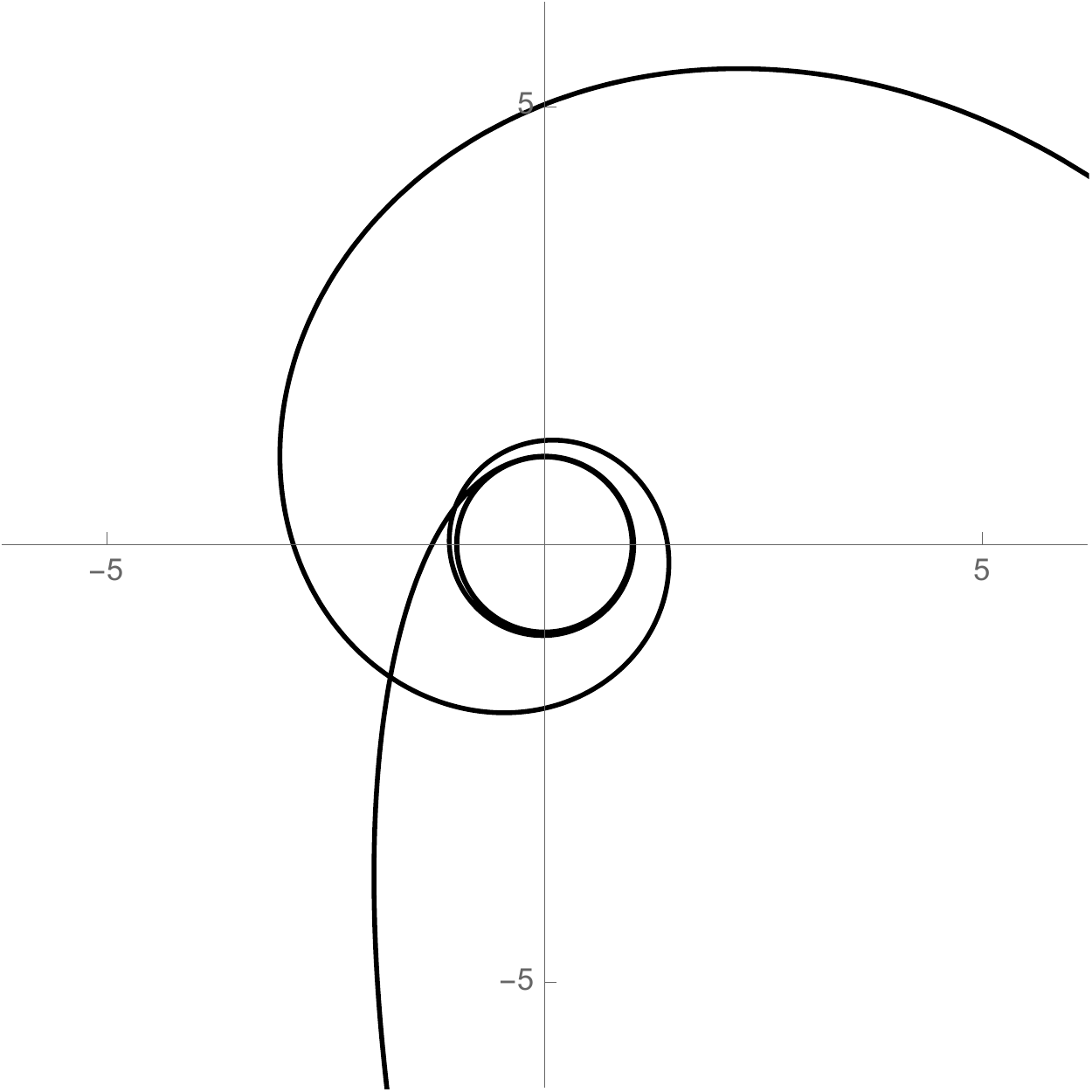}
\end{minipage}
\begin{minipage}{.3\linewidth}
\includegraphics[width=\linewidth]{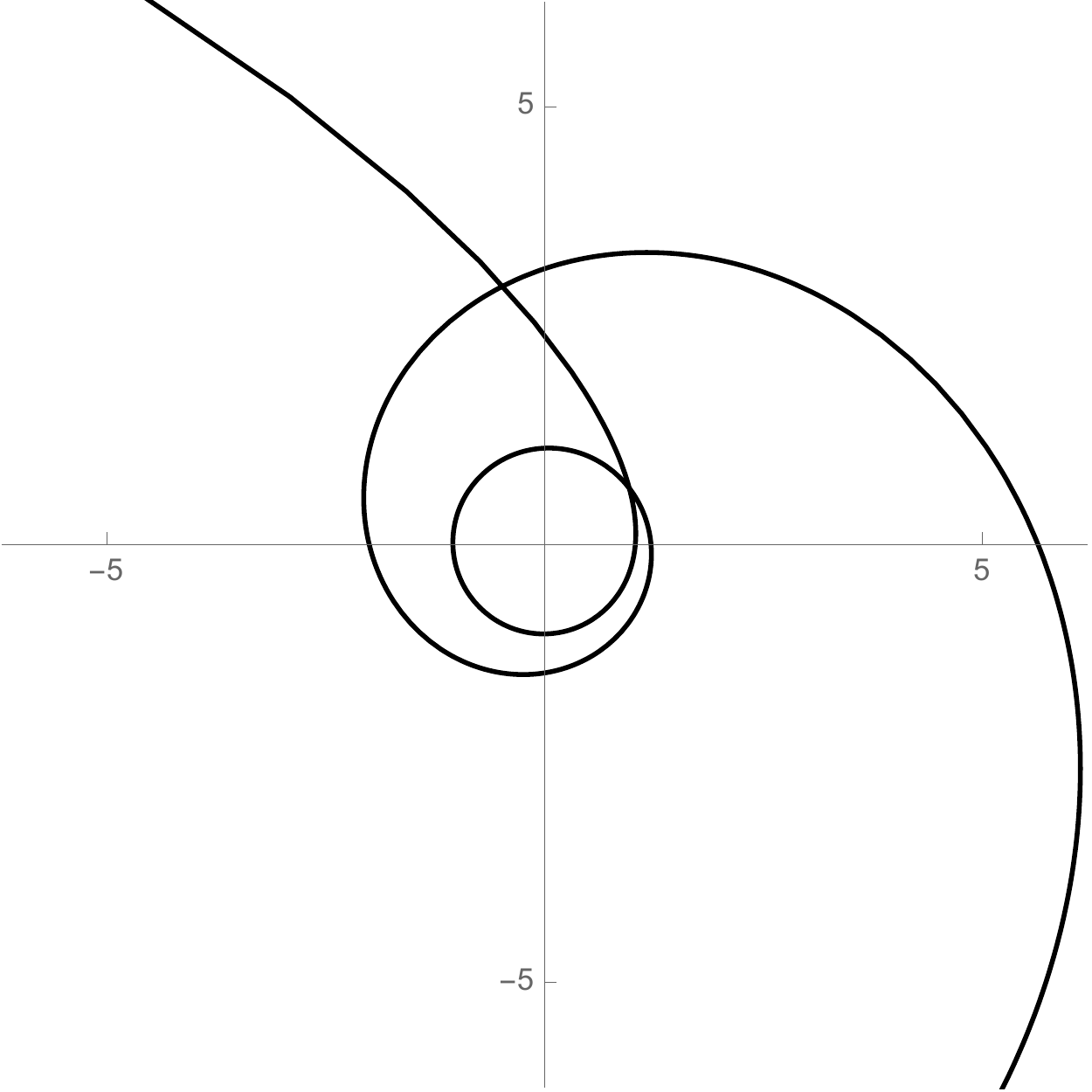}
\end{minipage}
\begin{minipage}{.3\linewidth}
\includegraphics[width=\linewidth]{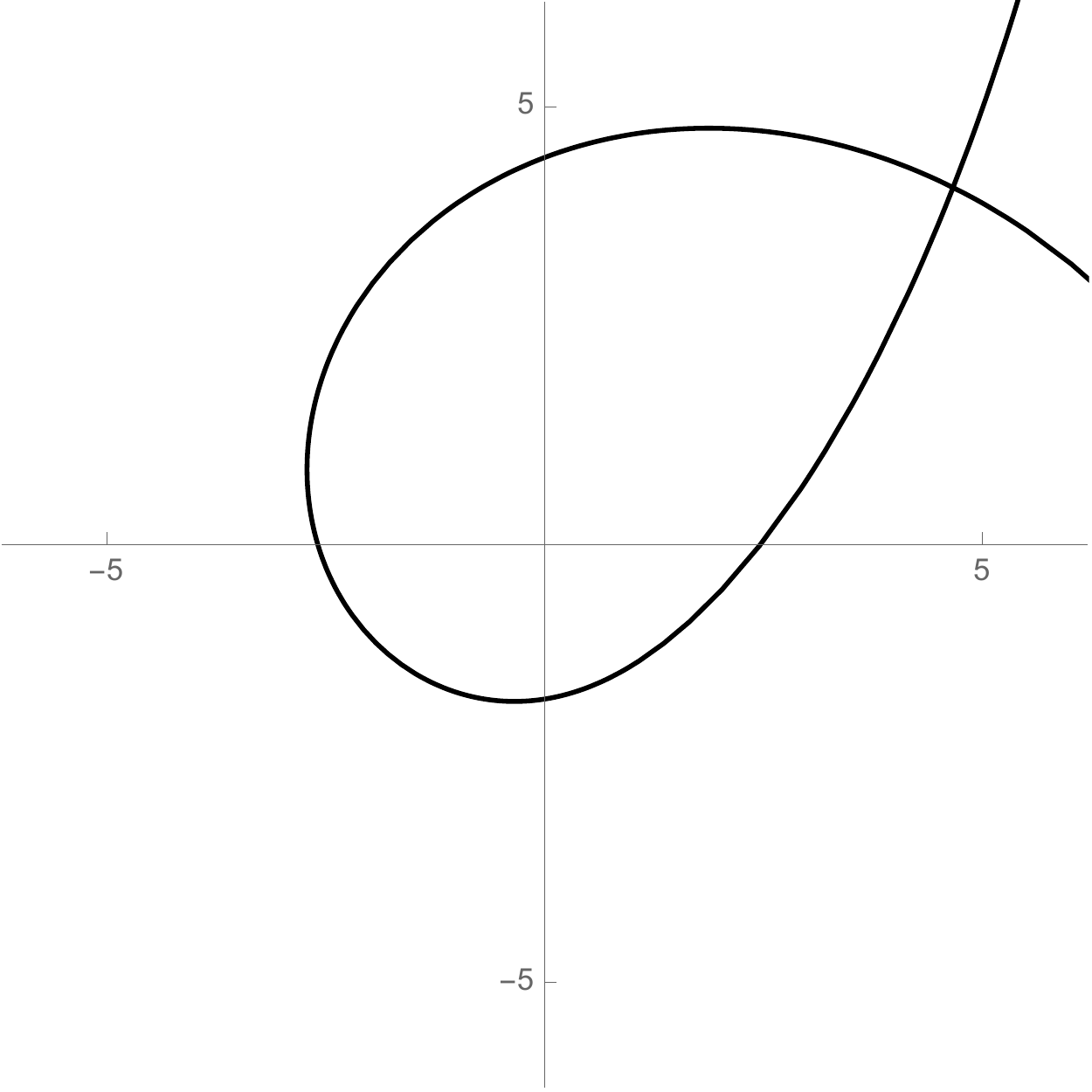}
\end{minipage}
\caption{Motion of plane curves $e^{- t_2} \gamma \left(x, t\right)$
corresponding to a 2-shock wave solution of the Burgers equation
\eqref{eqn:Burgers_with_a}. Parameters are $c = 1$,
$a = \pi/4$,
$\lambda_1 = - 1/2$,
$\lambda_2 = 4$,
$\xi_1 = \xi_2 = 0$ and 
$t_2=-12$ (left), $-2$ (middle), $-1/10$ (right).}
\label{fig:Burgers_curve-2}
\end{figure}
\end{center}

\section{Isogonal deformation of discrete curves}\label{sec:discrete_curve}
In this section, we consider the discrete deformation of discrete plane curves under the similarity
geometry, which naturally gives rise to the discrete Burgers equation and its hierarchy. For the 
definition and fundamental properties of the discrete Burgers hierarchy, the readers may refer to 
Appendix \ref{sec:discrete_Burgers}.

\subsection{Discrete curve}

For a map $\gamma\colon \Z \to \R^2,\, n \mapsto \gamma_n$, if any
consecutive three points $\gamma_{n + 1}, \gamma_n, \gamma_{n - 1}$ are
not colinear, we call $\gamma$ a {\em discrete plane curve}. For a
discrete plane curve $\gamma$, we denote by $q_n$ the distance between
the adjacent vertices
\begin{equation*}%\label{def:q}
q_n = \left|\gamma_{n + 1} - \gamma_n\right|.
\end{equation*}
We introduce $\kappa_n$ as the angle between the two vectors 
$\gamma_n - \gamma_{n - 1},\ \gamma_{n  + 1} - \gamma_n$.  
More precisely, we define $\kappa\colon \Z \to \left(0, 2 \pi\right)$ by
\begin{equation*}%\label{def:kappa}
\dfrac{\gamma_{n + 1} - \gamma_n}%
{q_n}
=
R \left(\kappa_n\right)
\dfrac{\gamma_n - \gamma_{n - 1}}%
{q_{n - 1}},
\end{equation*}
where $R$ is the rotation matrix
\begin{equation*}%\label{def:R}
R \left(x\right)
=
\begin{bmatrix}
\cos x & - \sin x\\
\sin x & \cos x
\end{bmatrix}.
\end{equation*}
Moreover, we put
\begin{equation*}%\label{def:frame}
T_n = \gamma_{n + 1} - \gamma_n,\quad
N_n = R \left(\dfrac{\pi}{2}\right) T_n,
\end{equation*}
and introduce the map $\phi\colon \Z \to \CO \left(2\right)$ by
\begin{equation*}%\label{eqn:discrete_frame}
\phi_n
= \left[T_n, N_n\right]
= q_n
\begin{bmatrix}
\dfrac{\gamma_{n + 1} - \gamma_n}{q_n},\ 
R \left(\dfrac{\pi}{2}\right)
\dfrac{\gamma_{n + 1} - \gamma_n}{q_n}
\end{bmatrix}.
\end{equation*}
We call the map $\phi$ the {\em similarity Frenet frame} of the discrete plane curve $\gamma$.
\begin{prop}\label{prop:frame-n}
The similairity Frenet frame $\phi$ satisfies the linear difference equation
\begin{equation}\label{eqn:X}
\phi_{n + 1} = \phi_n X_n,\quad
X_n = \dfrac{q_{n + 1}}{q_n} R \left(\kappa_{n + 1}\right).
\end{equation}
\end{prop}
\begin{proof}
Since $T$ satisfies
\begin{equation*}
\dfrac{1}{q_{n + 1}} T_{n + 1}
=
R \left(\kappa_{n + 1}\right)
\dfrac{1}{q_n} T_n,
\end{equation*}
we have
\begin{equation*}
\phi_{n + 1}
=
\left[T_{n + 1}, N_{n + 1}\right]
=
\dfrac{q_{n + 1}}{q_n}
R \left(\kappa_{n + 1}\right)
\left[T_n, N_n\right]
= X_n \phi_n.
\end{equation*}
Since the rotation matrix $R \left(\kappa_{n + 1}\right)$ and 
the matrix $\phi_n$ commute with each other, the statement is proved.
\end{proof}

\subsection{Isogonal Deformation}

\subsubsection{General settings}

We next consider the deformation of the curves. We write the deformed
curve as $\overline{\gamma}$, and we also express the data associated
with $\overline{\gamma}$ by putting $\overline{\phantom{\gamma}}$.
For instance, we define the function
$\overline{\kappa}\colon \Z \to \left(0, 2 \pi\right)$ by
\begin{equation}\label{def:overline-kappa}
\dfrac{\overline{\gamma}_{n + 1} - \overline{\gamma}_n}%
{\overline{q}_n}
=
R \left(\overline{\kappa}_n\right)
\dfrac{\overline{\gamma}_n - \overline{\gamma}_{n - 1}}%
{\overline{q}_{n-1}},\quad
\overline{q}_n=\left|\overline{\gamma}_{n + 1} - \overline{\gamma}_n\right|.
\end{equation}
\begin{lemma}\label{lemma:angle-preserving}
The necessary and sufficient condition for the deformation $\gamma
 \mapsto \overline{\gamma}$ being isogonal, namely, $\overline{\kappa} =
 \kappa$, is that there exist a positive-valued function $H$ and a
 constant $a$ satisfying
\begin{equation}\label{angle-preserving}
%\frac{\overline{\gamma}_{n + 1} - \overline{\gamma}_n}{\rho_n}
\overline{T}_{n}
=
H_n
\phi_n
\begin{bmatrix}
\cos a\\
\sin a
\end{bmatrix}.
\end{equation}
\end{lemma}
\begin{proof}
Since both $\overline{T}$ and $T$ are planar vectors,
it is obvious that there exist a positive-valued function $H$ and
an angle $a$ such that
\begin{equation*}
\overline{T}_n = H_n R \left(a_n\right) T_n
= H_n \phi_n
\begin{bmatrix}
\cos a_n\\
\sin a_n
\end{bmatrix}.
\end{equation*}
Therefore the equality $\overline{\kappa} = \kappa$ holds
if and only if the angle $a$ is independent of $n$.
\end{proof}
%%%%%%%%%%%%%%%%%%%%%%%%%%%%%%%%%%%%%%%%
% Proposition: 1step of deformation 
%%%%%%%%%%%%%%%%%%%%%%%%%%%%%%%%%%%%%%%%
\begin{prop}\label{prop:deformation1}
We fix $\delta\in\mathbb{R}_{>0}$, $a,\, f_0,\, g_0\in\mathbb{R}$ and a
 positive-valued function $H$.
We introduce the functions $f,\, g$ by the recursion relation
\begin{equation}\label{eqn:recursion_fg}
\begin{bmatrix}
f_{n + 1}\\
g_{n + 1}
\end{bmatrix}
=
\dfrac{1}{\delta}
\dfrac{q_n}{q_{n + 1}}
R \left(- \kappa_{n + 1}\right)
\begin{bmatrix}
1 + \delta f_n  - H_n \cos a\\
\delta g_n -  H_n \sin a
\end{bmatrix},
\end{equation}
and define the deformation $\gamma \mapsto \overline{\gamma}$ by
\begin{equation}\label{d-motion-1step}
\overline{\gamma}_n
=
\gamma_n
- \delta \left(f_n T_n + g_n N_n\right).
\end{equation}
Then we have the following:
\begin{enumerate}
\item
The deformation is isogonal. Namely, for the angle
$\overline{\kappa}_n$ defined by \eqref{def:overline-kappa}, we have
$\overline{\kappa}_n = \kappa_n$.

\item
The similarity Frenet frame $\overline{\phi}$ of the discrete curve $\overline{\gamma}$
can be expresed in terms the frame $\phi$ of $\gamma$ as
\begin{equation*}
\overline{\phi}_n
= \phi_n Y_n,\quad
Y_n = H_n R \left(a\right).
\end{equation*}
\end{enumerate}
\end{prop}
\begin{proof}
We compute the difference of $\overline{\gamma}$ by using
 \eqref{d-motion-1step}, \eqref{eqn:X} and
\eqref{eqn:recursion_fg}
\begin{align*}
\overline{T}_n =
\overline{\gamma}_{n+1} - \overline{\gamma}_n =\;&
\gamma_{n+1} - \delta \left(f_{n+1} T_{n+1} + g_{n+1} N_{n+1}\right)
- \gamma_{n} + \delta \left(f_{n} T_{n} + g_{n} N_{n}\right)\\
=\;& \left\{1 - \delta\frac{q_{n+1}}{q_n} \left(f_{n+1}\cos\kappa_{n+1} - g_{n+1} \sin\kappa_{n+1}\right) + \delta f_{n}\right\} T_n\\
&+ \delta \left\{- \frac{q_{n+1}}{q_n} \left(f_{n+1}\sin\kappa_{n+1}
+ g_{n+1} \cos\kappa_{n+1}\right) + g_n \right\} N_n\\
=\;& H_n \left(\cos a\, T_n + \sin a\, N_n\right).
\end{align*}
Then we have \eqref{angle-preserving}, which means
$\overline{\kappa} = \kappa$.
The frame of $\overline{\gamma}$ satisfies
\begin{equation*}
\overline{\phi}_n
= \left[\overline{T}_n,\,R \left(\dfrac{\pi}{2}\right) \overline{T}_n\right]
= H_n \phi_n R \left(a\right),
\end{equation*}
which completes the proof.
\end{proof}
Repeating the deformation in Proposition \ref{prop:deformation1},
we have the sequence of isogonal deformations of discrete plane curves
$\gamma^0 = \gamma$,
$\gamma^1 = \overline{\gamma}$, $\ldots$,
$\gamma^m = \overline{\gamma^{m - 1}}$, $\ldots$.
We write
\begin{gather*}
q_n^m = \left|\gamma_{n+1}^m-\gamma_n^m\right|,\\
%T^m_n = \dfrac{\gamma^m_{n + 1} - \gamma^m_n}{\rho_n},\quad
T^m_n = \gamma^m_{n + 1} - \gamma^m_n,\quad
N^m_n = R \left(\dfrac{\pi}{2}\right) T^m_n,\\
\dfrac{\gamma^m_{n + 1} - \gamma^m_n}{q_n^m}
= R \left(\kappa^m_n\right)
\dfrac{\gamma^m_n - \gamma^m_{n - 1}}{q_{n-1}^m}.
\end{gather*}
\begin{prop}
Let $\kappa$ be the angles associated with the discrete curve $\gamma^0$.  For each $m\in\Z$, we fix
$\delta_m>0$, real numbers $a_m,\, f^m_0,\, g^m_0$, and the positive-valued function $H^m$, and we
introduce the functions $f^m, g^m$ by the recursion relation
\begin{equation}\label{eqn:recursion_fg2}
\begin{split}
\begin{bmatrix}
f_{n + 1}^m\\
g_{n + 1}^m
\end{bmatrix}
&=
\dfrac{1}{\delta_m}
\dfrac{q_n^m}{q_{n + 1}^m}
R \left(- \kappa_{n + 1}\right)
\begin{bmatrix}
1 + \delta_m f_n^m  - H_n^m \cos a_m\\
\delta_m g_n^m -  H_n^m \sin a_m
\end{bmatrix}.
\end{split}
\end{equation}
Then defining the discrete curves $\gamma^m$ by
\begin{gather}
\gamma^{m + 1}_n
=\label{def:d-motion}
\gamma^m_n
- \delta_m \left(f^m_n T^m_n + g^m_n N^m_n\right),
\end{gather}
we have the following:
\begin{enumerate}
\item
For each $m$, it holds that $\kappa^m = \kappa^0 = \kappa$.
\item
The similarity Frenet frames $\phi^m$, $\phi^{m + 1}$ satisfy the system
of linear difference equations
\begin{align}
\phi^m_{n + 1}
&=\label{phi-L}
\phi^m_n X^m_n,\quad
X^m_n
= \dfrac{q^m_{n + 1}}{q^m_n}
R \left(\kappa_{n + 1}\right),\\
\phi^{m + 1}_n
&=\label{phi-M}
\phi^m_n Y^m_n,\quad
Y^m_n
= H^m_n R \left(a_m\right).
\end{align}
\item
The compatibility condition of the system of linear difference equation
\eqref{phi-L}--\eqref{phi-M} is
\begin{equation}\label{d-consistency}
\dfrac{q^{m + 1}_{n + 1}}{q^{m + 1}_n}
\dfrac{q^m_n}{q^m_{n + 1}}
=
\dfrac{H^m_{n + 1}}{H^m_n}.
\end{equation}
\end{enumerate}
\end{prop}

\subsubsection{Discrete Burgers flow}

Let us consider a special case where the function $\kappa$ is a constant
$\kappa_n = \epsilon$. For each $m\in\mathbb{Z}$, let $\delta_m$ be a
positive constant, and we set
\begin{gather}
a_m = 0,\quad \nonumber
f^m_0 = \dfrac{1}{\epsilon^2}\left(\dfrac{q^m_{- 1}}{q^m_0}-\cos \epsilon\right),
\quad 
g^m_0 = \dfrac{\sin \epsilon}{\epsilon^2},\\
H^m_n
=\label{def:H}
1 + \dfrac{\delta_m}{\epsilon^2}
\left(\dfrac{q^m_{n + 1}}{q^m_n}
- 2 \cos \epsilon
+ \dfrac{q^m_{n - 1}}{q^m_n}\right).
\end{gather}
Then the solution of the difference equation \eqref{eqn:recursion_fg2}
is given by
\begin{equation*}
f^m_n =  \frac{1}{\epsilon^2}\left(\dfrac{q^m_{n - 1}}{q^m_n} - \cos \epsilon\right)
,\quad
g^m_n = \dfrac{\sin \epsilon}{\epsilon^2}.
\end{equation*}
Defining the deformation of the discrete curve by \eqref{def:d-motion} by using this solution, the
compatibility condition \eqref{d-consistency} yields that the ratio $u_n^m = q_{n+1}^m/q_n^m$ obeys
(a variant of) the discrete Burgers equation (see \eqref{d-burgers-higher} with $i=2$)
\begin{equation}\label{eqn:autonomous_Burgers1}
\dfrac{u^{m + 1}_n}{u^m_n}
=
\dfrac{1 + \dfrac{\delta_m}{\epsilon^2}
\left(u^m_{n + 1} - 2 \cos \epsilon
+ \dfrac{1}{u^m_n}\right)}%
{1 + \dfrac{\delta_m}{\epsilon^2}
\left(u^m_n - 2 \cos \epsilon
+ \dfrac{1}{u^m_{n - 1}}\right)}.
\end{equation}
The length $q_n^m = \left|\gamma^m_{n + 1} - \gamma^m_n\right|$
satisfy the linear difference equation
\begin{equation*}%\label{eqn:linear_autonomous_Burgers1}
 \frac{q_{n}^{m+1}-q_n^m}{\delta_m} = \frac{q_{n+1}^m - 2 q_n^m \cos\epsilon + q_{n-1}^m}{\epsilon^2}.
\end{equation*}

\begin{rem}\label{rem:positivity_H1}
The function $H_n^m$ defined by \eqref{def:H} is not necessarily
 positive in general.
However, it is possible to make it positive by choosing $\delta_m>0$
 appropriately as follows.
We put
\begin{equation*}
Q_m = \min_n \dfrac{q^m_{n + 1} + q^m_{n - 1}}{q^m_n} = \min_n \left(u_{n+1}^m+\frac{1}{u_{n-1}^m}\right).
\end{equation*}
If $Q_m \geq 2 \cos \epsilon$, then we have for arbitrary $n$
\begin{equation*}
\dfrac{q^m_{n + 1} + q^m_{n - 1}}{q^m_n} - 2 \cos \epsilon
\geq Q_m - 2 \cos \epsilon
\geq 0,
\end{equation*}
which gives H$_n^m>0$. If $Q_m < 2 \cos \epsilon$, then choose $\delta_m$ as
\begin{equation*}
\dfrac{\epsilon^2}{2 \cos \epsilon - Q_m} > \delta_m>0,
\end{equation*}
then $H_n^m$ becomes positive. In fact, we have for arbitrary $n$
\begin{equation*}
H^m_n
=
1 + \dfrac{\delta_m}{\epsilon^2}
\left(\dfrac{q^m_{n + 1} + q^m_{n - 1}}{q^m_n}
- 2 \cos \epsilon\right)
\geq
1 + \dfrac{\delta_m}{\epsilon^2}
\left(Q_m - 2 \cos \epsilon\right)
> 0.
\end{equation*}
\end{rem}

\subsubsection{Discrete Burgers flow of higher order}

Let us write down the deformation equation corresponding to
\eqref{eqn:udot}. From \eqref{eqn:recursion_fg2} we have that
\begin{align*}
u_n^m f_{n+1}^m &= f_n^m\cos\kappa_{n+1} + g_n^m\sin\kappa_{n+1} + \frac{\cos\kappa_{n+1} - H_n^m\cos(\kappa_{n+1}-a_m)}{\delta_m},\\
u_n^m g_{n+1}^m &= -f_n^m\sin\kappa_{n+1} + g_n^m\cos\kappa_{n+1} - \frac{\sin\kappa_{n+1} - H_n^m\sin(\kappa_{n+1}-a_m)}{\delta_m}.
\end{align*}
We solve the second equation in terms of $f_n^m$ and substitute it into
the first equation with $n\mapsto n-1$ so as to obtain that
\begin{equation*}%\label{eqn:non_autonomous_linear_eq_example}
\begin{split}
\frac{\sin(\kappa_{n+1}-a_m)}{\sin\kappa_{n+1}}H_n^m
&+ \frac{\sin a_m}{\sin\kappa_{n}} \frac{1}{u_{n-1}^m}H_{n-1}^m\\
= 1 + \delta_m &\left\{\frac{1}{\sin\kappa_{n+1}}u_n^m g_{n+1}^m 
-\left(\frac{\cos\kappa_{n+1}}{\sin\kappa_{n+1}} + \frac{\cos\kappa_{n} }{\sin\kappa_{n}} \right) g_n^m + \frac{1}{\sin\kappa_{n}} \frac{1}{u_{n-1}^m}g_{n-1}^m\right\}.
\end{split}
\end{equation*}
Then the compatibility condition \eqref{d-consistency} gives
\begin{equation}\label{eqn:u_g1}
\begin{minipage}{0.85\textwidth}
%\begin{footnotesize}
\begin{displaymath}
\begin{split} 
&\dfrac{u_n^{m+1}}{u_n^m}\frac{\dfrac{\sin(\kappa_{n+2}-a_m)}{\sin\kappa_{n+2}}  + \dfrac{\sin a_m}{\sin\kappa_{n+1}} \dfrac{1}{u_{n}^{m+1}}}
{\dfrac{\sin(\kappa_{n+1}-a_m)}{\sin\kappa_{n+1}}  + \dfrac{\sin a_m}{\sin\kappa_{n}} \dfrac{1}{u_{n-1}^{m+1}}}\nonumber\\[2mm]
&\hskip40pt 
=\frac{1 +  \delta_m\left\{\dfrac{1}{\sin\kappa_{n+2}}u_{n+1}^m g_{n+2}^m 
- \left(\dfrac{\cos\kappa_{n+2}}{\sin\kappa_{n+2}} + \dfrac{\cos\kappa_{n+1} }{\sin\kappa_{n+1}} \right) g_{n+1}^m 
 + \dfrac{1}{\sin\kappa_{n+1}} \dfrac{1}{u_{n}^m}g_{n}^m\right\}}
{1 +  \delta_m\left\{\dfrac{1}{\sin\kappa_{n+1}}u_n^m g_{n+1}^m 
- \left(\dfrac{\cos\kappa_{n+1}}{\sin\kappa_{n+1}} + \dfrac{\cos\kappa_{n} }{\sin\kappa_{n}} \right) g_n^m 
 + \dfrac{1}{\sin\kappa_{n}} \dfrac{1}{u_{n-1}^m}g_{n-1}^m\right\}},
\end{split} 
\end{displaymath}
%\end{footnotesize}
\end{minipage}
\end{equation}
or equivalently
\begin{equation}\label{eqn:u_g2}
\begin{minipage}{0.85\textwidth}
\begin{footnotesize}
\begin{displaymath}
\begin{split}
&\dfrac{u_n^{m+1}}{u_n^m}\frac{\dfrac{\sin(\kappa_{n+2}-a_m)}{\sin\kappa_{n+2}}  + \dfrac{\sin a_m}{\sin\kappa_{n+1}} \dfrac{1}{u_{n}^{m+1}}}
{\dfrac{\sin(\kappa_{n+1}-a_m)}{\sin\kappa_{n+1}}  + \dfrac{\sin a_m}{\sin\kappa_{n}} \dfrac{1}{u_{n-1}^{m+1}}}\nonumber\\[2mm]
&
=
\frac{1 +  \delta_m\left(\dfrac{1}{\sin\kappa_{n+2}}u_{n+1}^m e^{\partial_n}
- \dfrac{1}{\sin\kappa_{n+2}} - \dfrac{1}{\sin\kappa_{n+1}}
 + \dfrac{1}{\sin\kappa_{n+1}} \dfrac{1}{u_{n}^m}e^{-\partial_n}\right)g_{n+1}^m  + \delta_m\left(\tan\dfrac{\kappa_{n+2}}{2} + \tan\dfrac{\kappa_{n+1}}{2}\right)g_{n+1}^m}
{1 +  \delta_m\left(\dfrac{1}{\sin\kappa_{n+1}}u_n^m e^{\partial_n}
- \dfrac{1}{\sin\kappa_{n+1}} - \dfrac{1}{\sin\kappa_{n}}
 + \dfrac{1}{\sin\kappa_{n}} \dfrac{1}{u_{n-1}^m}e^{-\partial_n}\right) g_n^m  + \delta_m\left(\tan\dfrac{\kappa_{n+1}}{2} + \tan\dfrac{\kappa_{n}}{2}\right)g_{n}^m}.
\end{split}
\end{displaymath}
\end{footnotesize}
\end{minipage}
\end{equation}
Equation \eqref{eqn:u_g1} or \eqref{eqn:u_g2} is the general form of the
deformation equation of the discrete curves in the framework of the
similarity geometry, and is regarded as a discrete counterpart of
\eqref{eqn:udot}.
\begin{thm}\label{thm:discrete_deformation}
For a fixed $m\in\mathbb{Z}$, let $\gamma_n^m\in\mathbb{R}^2$ be a
 discrete curve, and let $\kappa_n=\angle(\gamma_{n+1}^m-\gamma_n^m,
 \gamma_n^m-\gamma_{n-1}^m)$, $q_n^m= |\gamma_{n+1}^m-\gamma_n^m|$,
 $u_n^m= \frac{q_{n+1}^m}{q_n^m}$. For given $\delta_m, a_m,
 H_0^m\in\mathbb{R}$ and a function $g_n^m\in\mathbb{R}$, we define
 $H_n^m\in\mathbb{R}$ recursively by
\begin{equation}\label{eqn:determine_H}
\begin{split}
& \frac{\sin(\kappa_{n+1}-a_m)}{\sin\kappa_{n+1}}H_n^m  + \frac{\sin a_m}{\sin\kappa_{n}}
 \frac{1}{u_{n-1}^m}H_{n-1}^m\\
&
= 
1+  \delta_m\left\{\frac{u_n^m g_{n+1}^m }{\sin\kappa_{n+1}}
-\left(\frac{1}{\sin\kappa_{n+1}} + \frac{1}{\sin\kappa_{n}} \right) g_n^m 
 + \frac{1}{\sin\kappa_{n}} \frac{g_{n-1}^m}{u_{n-1}^m}\right\}
+ \delta_m\left(\tan\frac{\kappa_{n+1}}{2}+\tan\frac{\kappa_n}{2}\right)g_n^m.
\end{split}
\end{equation}
Then we have:
\begin{enumerate}
 \item By choosing $\delta_m$ and $a_m$ appropriately, $H_n^m$ becomes positive.
 \item Setting the function $f_n^m$ by
\begin{equation*}
  f_n^m = - \frac{u_n^m}{\sin\kappa_{n+1}} g_{n+1}^m + \frac{\cos\kappa_{n+1}}{\sin\kappa_{n+1}}g_n^m  
- \frac{1}{\delta_m}   + \frac{\sin(\kappa_{n+1}-a_m)}{\delta_m\sin\kappa_{n+1}}H_n^m,
\end{equation*}
the condition \eqref{eqn:recursion_fg2} is satisfied. Namely,
\eqref{d-motion-1step} gives a isogonal deformation.
 \item $u_n^m$ satisfies \eqref{eqn:u_g2}.
\end{enumerate}
\end{thm}

%%%%%%%%%%%%%%%%%%%%%%%%%%%%%%%%%%%%%%%%%%%%%%%%%%%%%%%%%%%%%
%
%%%%%%%%%%%%%%%%%%%%%%%%%%%%%%%%%%%%%%%%%%%%%%%%%%%%%%%%%%%%%
We note that \eqref{eqn:u_g2} yields the discrete Burgers equation and
its generalizations to that of higher-order by suitable specialization
of $g_n^m$.

\paragraph{Autonomous case}

In the case
of $\kappa_n=\epsilon=\mathrm{const.}$, \eqref{eqn:u_g2} is reduced to
\begin{align}
\dfrac{u_n^{m+1}}{u_n^m}\frac{\sin(\epsilon-a_m) + \dfrac{\sin a_m }{u_{n}^{m+1}}}{\sin(\epsilon-a_m) + \dfrac{\sin a_m}{u_{n-1}^{m+1}}}
&=\label{eqn:u_g_autonomous1-1}
\frac{1 +  \dfrac{\delta_m}{\sin\epsilon}\left(u_{n+1}^m e^{\partial_n}
- 2 + \dfrac{1}{u_{n}^m}e^{-\partial_n}\right)g_{n+1}^m  + \delta_m\dfrac{2-2\cos\epsilon}{\sin\epsilon} g_{n+1}^m}{1 +  \dfrac{\delta_m}{\sin\epsilon} \left(u_n^m e^{\partial_n}- 2 + \dfrac{1}{u_{n-1}^m}e^{-\partial_n}\right) g_n^m
+ \delta_m\dfrac{2-2\cos\epsilon}{\sin\epsilon}g_{n}^m}\\
&=\label{eqn:u_g_autonomous1-2}
\frac{1 + \delta_m \dfrac{\epsilon^2}{\sin\epsilon}\Omega_{n+1}^{(2)}g_{n+1}^m + \delta_m\dfrac{2-2\cos\epsilon}{\sin\epsilon} g_{n+1}^m}{1 + \delta_m \dfrac{\epsilon^2}{\sin\epsilon}\Omega_n^{(2)} g_n^m  + \delta_m\dfrac{2-2\cos\epsilon}{\sin\epsilon}g_{n}^m},
\end{align}
where $\Omega_n^{(2)}$ is the recursion operator of the discrete
Burgers hierarchy given in \eqref{eqn:recursion_autonomous2}.
Putting $g_n=\frac{\sin \epsilon}{\epsilon^2}$ and $a_m=0$,
\eqref{eqn:u_g_autonomous1-1} recovers the autonomous discrete Burgers
equation \eqref{eqn:autonomous_Burgers1}.
Equation \eqref{eqn:u_g_autonomous1-2} is a discrete counterpart of
\eqref{eqn:udot}.
Therefore, due to \eqref{eqn:recursion_autonomous3},
by putting $g_n^m$ as
\begin{equation*}
g_n^m=\dfrac{\sin \epsilon}{\epsilon^2}
\hat{K}^{(i)}[u_n^m],
\end{equation*}
we obtain a variant of the higher order autonomous discrete Burgers
equation
\begin{equation*}%\label{eqn:deformation_dBurgers_auto}
\dfrac{u_n^{m+1}}{u_n^m} \frac{1 + \dfrac{\sin a_m}{\sin(\epsilon -a_m)} \dfrac{1}{u_{n}^{m+1}}}
{1  + \dfrac{\sin a_m}{\sin(\epsilon - a_m)} \dfrac{1}{u_{n-1}^{m+1}}}
= 
\frac{1 +  \delta_m\hat{K}^{(i+2)}[u_{n+1}^m]  + \dfrac{2-2\cos \epsilon}{\epsilon^2}\delta_m\hat{K}^{(i)}[u_{n+1}^m]}
{1 +  \delta_m\hat{K}^{(i+2)}[u_{n}^m]   + \dfrac{2-2\cos \epsilon}{\epsilon^2}\delta_m\hat{K}^{(i)}[u_n^m]},
\end{equation*}
which corresponds to \eqref{eqn:n-th_Burgers}.

\paragraph{Non-autonomous case}

For the case of generic $\kappa_n$, we see that the recursion operator
of the non-autonomous discrete Burgers hierarchy appears in the right
hand side of \eqref{eqn:u_g2}. In fact, we have
\begin{equation*}
 \dfrac{1}{\sin\kappa_{n+1}}u_n^m e^{\partial_n}
- \dfrac{1}{\sin\kappa_{n+1}} - \dfrac{1}{\sin\kappa_{n}}
 + \dfrac{1}{\sin\kappa_{n}} \dfrac{1}{u_{n-1}^m}e^{-\partial_n}
= \epsilon_n^{(i+2)}\Omega_n^{(2,i+2)},
\end{equation*}
by parametrizing $\sin \kappa_n$ as
\begin{equation*}
\sin \kappa_n =
\begin{cases}
\epsilon_{n-1}^{(i+1)} & i=2l,\\
\epsilon_n^{(i+1)} & i=2l+1,
\end{cases}
\end{equation*}
where $\epsilon_n^{(i)}$ and $\Omega_n^{(i)}$ are given in
\eqref{eqn:nonautonomous_epsilon} and
\eqref{eqn:recursion_nonautonomous_Burgers}, respectively.
For the simplest case $i = 0$,
we choose $g_n^m=1$ in \eqref{eqn:u_g2} and have
\begin{equation*}
 \dfrac{u_n^{m+1}}{u_n^m} \frac{\dfrac{\sin(\kappa_{n+2}-a_m)}{\sin\kappa_{n+2}}  + \dfrac{\sin a_m}{\sin\kappa_{n+1}} \dfrac{1}{u_{n}^{m+1}}}
{\dfrac{\sin(\kappa_{n+1}-a_m)}{\sin\kappa_{n+1}}  + \dfrac{\sin a_m}{\sin\kappa_{n}} \dfrac{1}{u_{n-1}^{m+1}}}
= 
\frac{1 +  \delta_m \epsilon_{n+1}^{(2)}K_{n+1}^{(2)}[u_{n+1}^m]  + \delta_m\left(\tan\dfrac{\kappa_{n+2}}{2} + \tan\dfrac{\kappa_{n+1}}{2}\right)}
{1 +  \delta_m \epsilon_{n}^{(2)} K_n^{(2)}[u_{n}^m]   + \delta_m\left(\tan\dfrac{\kappa_{n+1}}{2} + \tan\dfrac{\kappa_{n}}{2}\right)},
\end{equation*}
which is a non-autonomous discrete analogue of the Burgers equation
\eqref{eqn:Burgers_with_a}. If we set $a_m=0$, we obtain a simpler
version of the non-autonomous discrete Burgers equation
\begin{equation*}
\frac{u_n^{m+1}}{u_n^m}= 
\frac{1 +  \delta_m\epsilon_{n+1}^{(2)}K_{n+1}^{(2)}[u_{n+1}^m]  + \delta_m\left(\tan\dfrac{\kappa_{n+2}}{2} + \tan\dfrac{\kappa_{n+1}}{2}\right)}
{1 + \delta_m \epsilon_{n}^{(2)}K_n^{(2)}[u_{n}^m]   + \delta_m\left(\tan\dfrac{\kappa_{n+1}}{2} + \tan\dfrac{\kappa_{n}}{2}\right)}.
\end{equation*}
For $i > 0$, we put $g_n^m=K_n^{(i)}[u_n^m]$ and find that $u_n^m$
satisfies a variant of non-autonomous higher-order discrete Burgers
equation
\begin{equation}\label{eqn:deformation_dBurgers}
\begin{split}
\dfrac{u_n^{m+1}}{u_n^m} &
\frac{\dfrac{\sin(\kappa_{n+2}-a_m)}{\sin\kappa_{n+2}}  + \dfrac{\sin a_m}{\sin\kappa_{n+1}} \dfrac{1}{u_{n}^{m+1}}}
{\dfrac{\sin(\kappa_{n+1}-a_m)}{\sin\kappa_{n+1}}  + \dfrac{\sin a_m}{\sin\kappa_{n}} \dfrac{1}{u_{n-1}^{m+1}}}\\
&= 
\frac{1 + \delta_m \epsilon_{n+1}^{(i+2)}K_{n+1}^{(i+2)}[u_{n+1}^m]  + \delta_m\left(\tan\dfrac{\kappa_{n+2}}{2} + \tan\dfrac{\kappa_{n+1}}{2}\right)K_{n+1}^{(i)}[u_{n+1}^m]}
{1 + \delta_m \epsilon_{n}^{(i+2)}K_n^{(i+2)}[u_{n}^m] +
 \delta_m\left(\tan\dfrac{\kappa_{n+1}}{2} +
 \tan\dfrac{\kappa_{n}}{2}\right)K_n^{(i)}[u_n^m]},
\end{split}
\end{equation}
which is a non-autonomous discrete analogue of the higher-order Burgers
equation \eqref{eqn:n-th_Burgers}. Note that $q_n^m=|\gamma_{n+1}^m -
\gamma_n^m|$ satisfies the linear equation
\begin{equation*}%\label{eqn:deformation_dBurgers_linear}
\frac{1}{\delta_m}
\left\{\dfrac{\sin(\kappa_{n+1}-a_m)}{\sin\kappa_{n+1}} q_n^{m+1}  + \dfrac{\sin a_m}{\sin\kappa_{n}} q_{n-1}^{m+1} - q_n^m\right\}
=\epsilon_{n}^{(i+2)}L_n^{(i+2)}[q_{n}^m]   + \left(\tan\dfrac{\kappa_{n+1}}{2} + \tan\dfrac{\kappa_{n}}{2}\right)L_n^{(i)}[q_n^m].
\end{equation*}

%%%%%%%%%%%%%%%%%%%%%%%%%%%%%%%%%%%%%%%%%%%%%%%%%%
%
%%%%%%%%%%%%%%%%%%%%%%%%%%%%%%%%%%%%%%%%%%%%%%%%%%
%\begin{proof}[Proof of Theorem \ref{thm:discrete_deformation}]

We now prove Theorem \ref{thm:discrete_deformation}. The statement (2)
and (3) are derived immediately by solving \eqref{eqn:recursion_fg2} and
using the compatibility condition \eqref{d-consistency}. For the
statement (1), we have the following as a sufficient condition for the
positivity of $H_n^m$:
%%%%%%%%%%%%%%%%%%%%%%%%
%
%%%%%%%%%%%%%%%%%%%%%%%%
\begin{lemma}\label{lem:condition_delta_a}
We assume that $\kappa_n$ satisfies $0<\kappa_n<\pi$ or
$-\pi<\kappa_n<0$ for all $n$.
For each $m$, we choose $\delta_m$ and $a_m$ in the following manner:
\begin{equation}\label{eqn:condition_delta_a}
\left\{\begin{array}{ll}
{\displaystyle 0<\delta_m}& (U^m_{\rm min}>0)\\
{\displaystyle  0<\delta_m<-1/U^m_{\rm min}}
& (U^m_{\rm min}<0)
\end{array}\right. , \quad
\left\{\begin{array}{ll}
{\displaystyle \kappa_{\rm max}-\pi<a_m<0}& (0<\kappa_n<\pi)\\
 {\displaystyle  0<a_m<\kappa_{\rm min}+\pi}
& (-\pi<\kappa_n<0)
\end{array}\right. ,
\end{equation}
%\begin{equation}\label{eqn:condition_a}
%
%\end{equation}
where
\begin{equation*}
 U_n^m=\frac{u_n^m g_{n+1}^m }{\sin\kappa_{n+1}}
-\left(\frac{1}{\sin\kappa_{n+1}} + \frac{1}{\sin\kappa_{n}} \right) g_n^m 
 + \frac{1}{\sin\kappa_{n}} \frac{g_{n-1}^m}{u_{n-1}^m}
+ \left(\tan\frac{\kappa_{n+1}}{2}+\tan\frac{\kappa_n}{2}\right)g_n^m,
\end{equation*}
and 
\begin{equation*}
U_{{\rm max}}^{m}=\max_{n} U_{n}^{m},\quad
U_{{\rm min}}^m=\min_{n} U_{n}^m,\quad
\kappa_{\rm min}=\min\limits_{n}\kappa_n,\quad
\kappa_{\rm max}=\max\limits_{n}\kappa_n.
\end{equation*}
Then we have $H_n^m>0$.
\end{lemma}
%%%%%%%%%%%%%%%%%%%%%%%%%%%%%%%
%
%%%%%%%%%%%%%%%%%%%%%%%%%%%%%%%
\begin{proof}[Proof of Lemma \ref{lem:condition_delta_a}]
We first write the recursion relation \eqref{eqn:determine_H} as
\begin{equation}\label{eqn:eq_for_H}
H_n^{m} = -\alpha_n^m H_{n-1}^m
+ \beta_n^m\left(1+\delta_m U_{n}^{m}\right),
\end{equation}
where
\begin{equation*}
\alpha_n^m = \dfrac{\sin a_m\sin\kappa_{n+1}}{u_{n-1}^m\sin(\kappa_{n+1}-a_m)\sin\kappa_{n}},\quad
\beta_n^m=\dfrac{\sin\kappa_{n+1}}{\sin(\kappa_{n+1}-a_m)}.
\end{equation*}
Then \eqref{eqn:eq_for_H} can be solved formally as
\begin{equation*}
 H_n^m =\left( H_0^m + \sum_{\nu=0}^n \beta_\nu^m (1 + \delta_m U_{\nu}^{m})\prod_{k=0}^\nu (-\alpha_k^m)^{-1}\right)
\prod_{\mu=0}^{n}(-\alpha_\mu^m),\quad H_0^m>0.
\end{equation*}
Noticing that $\delta_m, u_n^m>0$, it is sufficient for $H_n^m>0$ that all
of the following conditions
\begin{align}
&\dfrac{\sin\kappa_{n+1}}{\sin(\kappa_{n+1}-a_m)}>0,\label{eqn:sufficient_1}\\
& 1 + \delta_m U_{n}^{m} > 0 ,\label{eqn:sufficient_2}\\
&\prod_{\nu=0}^n \left(-\dfrac{\sin a_m\sin\kappa_{\nu+1}}{\sin(\kappa_{\nu+1}-a_m)\sin\kappa_{\nu}}\right)>0,
\label{eqn:sufficient_3}
\end{align}
are satisfied for all $n$. Then it is easy to see that
\eqref{eqn:sufficient_2} is satisfied by choosing $\delta_m$ as
\eqref{eqn:condition_delta_a}. The conditions \eqref{eqn:sufficient_1}
and \eqref{eqn:sufficient_2} imply
%\begin{equation}
%\begin{split}
%& \text{(i)}\quad \dfrac{\sin a_m}{\sin(\kappa_{\nu+1}-a_m)}<0\quad \text{and}\quad 
%\dfrac{\sin\kappa_{\nu+1}}{\sin\kappa_{\nu}}>0\quad \text{for}\  {}^\forall n\\
%&\text{or}\\
%&\text{(ii)}\quad \dfrac{\sin a_m}{\sin(\kappa_{\nu+1}-a_m)}>0\quad \text{and}\quad 
%\dfrac{\sin\kappa_{\nu+1}}{\sin\kappa_{\nu}}<0\quad \text{for}\  {}^\forall n
%\end{split}
%\end{equation}
%For case (i),  \eqref{eqn:sufficient_1} 
\begin{equation}
\begin{split}
& \sin\kappa_n>0,\quad\sin(\kappa_n-a_m)>0 ,\quad \sin a_m<0\quad \text{for}\ {}^\forall n, \\
\text{or}& \\
& \sin\kappa_n<0,\quad \sin(\kappa_n-a_m)<0 ,\quad\sin a_m>0\quad \text{for}\ {}^\forall n,
\end{split}
\end{equation}
from which we have
\begin{equation}\label{eqn:condition_a1}
\begin{split}
& 0<\kappa_n<\pi,\quad \kappa_{\max}- a_m< \pi,\quad -\pi<a_m<0,\\ 
\text{or}&\\
& -\pi<\kappa_n<0,\quad -\pi < \kappa_{\min}- a_m,\quad 0<a_m<\pi.
\end{split}
\end{equation}
This is equivalent to the second condition in \eqref{eqn:condition_delta_a}. 
\end{proof}

\subsection{Explicit formula}
An explicit representation formula for the curve $\gamma_n^m$ is
constructed in a similar manner to the smooth curves.
\begin{prop}
Let $\gamma_n^m$ be a discrete curve satisfying \eqref{phi-L} and
\eqref{phi-M}.  Then $\gamma_n^m$ admits the representation formula
\begin{equation}\label{eqn:discrete_representation}
\gamma^m_n
=
\sum_{j}^{n-1}
q^m_j
\begin{bmatrix}
\cos \theta^m_j\\
\sin \theta^m_j
\end{bmatrix},\quad
\theta^m_n
=
\sum_{\nu}^n \kappa_\nu + \sum_{\mu}^m a_\mu.
\end{equation}
\end{prop}
\begin{proof}
 Since $|\gamma_{n+1}^m - \gamma_n^m|/q_n^m = 1$, there exist a function $\theta_n^m\in[0,2\pi)$ such that
\begin{equation}
 \frac{\gamma_{n+1}^m - \gamma_n^m}{q_n^m} 
= \left[\begin{array}{c}\cos\theta_n^m \\ \sin\theta_n^m \end{array}\right],
\end{equation}
so that the frame $\phi_n^m$ is expressed as %\textcolor{red}{（$\rho_n$とは？）}
\begin{equation}
% \phi_n^m = \frac{q_n^m}{\rho_n}R(\theta_n^m).
\phi_n^m = q_n^m R(\theta_n^m).
\end{equation}
Then \eqref{phi-L} and \eqref{phi-M} give
\begin{equation}
 \theta_{n+1}^m - \theta_n^m - \kappa_n\in2\pi\mathbb{Z},\quad 
 \theta_{n}^{m+1} - \theta_n^m - a_m\in2\pi\mathbb{Z},
\end{equation}
and we may assume
\begin{equation}
 \theta_{n+1}^m - \theta_n^m - \kappa_n = 0,\quad
 \theta_{n}^{m+1} - \theta_n^m - a_m = 0,
\end{equation}
without losing generality, which implies \eqref{eqn:discrete_representation}.
\end{proof}

For the curves constructed from the shock wave solutions of the autonomous discrete Burgers
hierarchy, the summation in \eqref{eqn:discrete_representation} can be computed explicitly.  For simplicity, we
demonstrate it by taking the case of $i=2$ with $\kappa_n=\epsilon$ (const.), $\delta_m=\delta$
(const.) in \eqref{eqn:deformation_dBurgers}
\begin{equation}\label{eqn:auto_dBurgers_with_a}
 \dfrac{u_n^{m+1}}{u_n^m} \frac{\dfrac{\sin(\epsilon-a_m)}{\sin\epsilon}  + \dfrac{\sin a_m}{\sin\epsilon} \dfrac{1}{u_{n}^{m+1}}}
{\dfrac{\sin(\epsilon-a_m)}{\sin\epsilon}  + \dfrac{\sin a_m}{\sin\epsilon} \dfrac{1}{u_{n-1}^{m+1}}}
= 
\dfrac{1 + \dfrac{\delta_m}{\epsilon^2}
\left(u^m_{n + 1} - 2 \cos \epsilon
+ \dfrac{1}{u^m_n}\right)}%
{1 + \dfrac{\delta_m}{\epsilon^2}
\left(u^m_n - 2 \cos \epsilon
+ \dfrac{1}{u^m_{n - 1}}\right)},
\end{equation}
which is linearized in terms of $q_n^m$ as
\begin{equation}\label{eqn:linear_auto_dBurgers_with_a}
 \frac{\dfrac{\sin(\epsilon-a_m)}{\sin\epsilon} q_n^{m+1}  + \dfrac{\sin a_m}{\sin\epsilon} q_{n-1}^{m+1} - q_n^m}{\delta} 
=\frac{q_{n+1}^m - 2 q_n^m \cos\epsilon + q_{n-1}^m}{\epsilon^2}.
\end{equation}
\eqref{eqn:linear_auto_dBurgers_with_a} admits the solution
\begin{equation}
q_n^m = e^{\mu_0 m} +
\sum_{k=1}^M \exp(\lambda_k n + \mu_k m + \xi_k),
\end{equation}
where $M\in\mathbb{N}$, $\lambda_k$, $\xi_k$ ($k=1,\ldots,M$) are arbitrary constants and
\begin{equation}
 \mu_k = \log \frac{1+\frac{\delta}{\epsilon^2}\left(e^{\lambda_k} - 2\cos\epsilon + e^{-\lambda_k}\right)}
{\frac{\sin(\epsilon-a_m)}{\sin\epsilon} + \frac{\sin a_m}{\sin\epsilon}e^{-\lambda_k} }.
\end{equation}
Then, by using the formulas
\begin{align*}
\sum_{j}^{n - 1}
c^j \cos \left(j \epsilon\right)
&=
\dfrac{c^{n + 1} \cos \left(\left(n - 1\right) \epsilon\right)
- c^n \cos \left(n \epsilon\right)}%
{c^2 - 2 c \cos \epsilon + 1},\\
\sum_{j}^{n - 1}
c^j \sin \left(j \epsilon\right)
&=
\dfrac{c^{n + 1} \sin \left(\left(n - 1\right) \epsilon\right)
- c^n \sin \left(n \epsilon\right)}%
{c^2 - 2 c \cos \epsilon + 1},
\end{align*}
where $c$ is a constant satisfying $c^j\to 0$ ($j\to -\infty$), we have
from \eqref{eqn:discrete_representation} that
\begin{equation}
\gamma^m_n=
\sum_{k = 0}^M
\dfrac{\exp \left(\lambda_k n + \mu_k m + \xi_k\right)}%
{e^{2 \lambda_k} - 2 e^{\lambda_k} \cos \epsilon + 1}
\begin{bmatrix}
e^{\lambda_k}
\cos \theta_{n-1}^m
- \cos \theta_n^m\\
e^{\lambda_k}
\sin \theta_{n-1}^m
- \sin \theta_n^m
\end{bmatrix},
\end{equation}
with $\lambda_0 = \xi_0 = 0$ and $\lambda_k>0$.
Figure \ref{fig:dBurgers_curve-1}, \ref{fig:dBurgers_curve-2} illustrate
motion of discrete plane curves corresponding to $M$-shock wave
solutions ($M = 1, 2$, respectively) of the discrete Burgers equation
\eqref{eqn:auto_dBurgers_with_a} with parameters $a = \pi/3$, $\epsilon
= \pi/4$, $\delta = 1$, $\xi_1 = 0$.
\begin{center}
\begin{figure}[h]
\begin{minipage}{.3\linewidth}
\includegraphics[width=\linewidth]{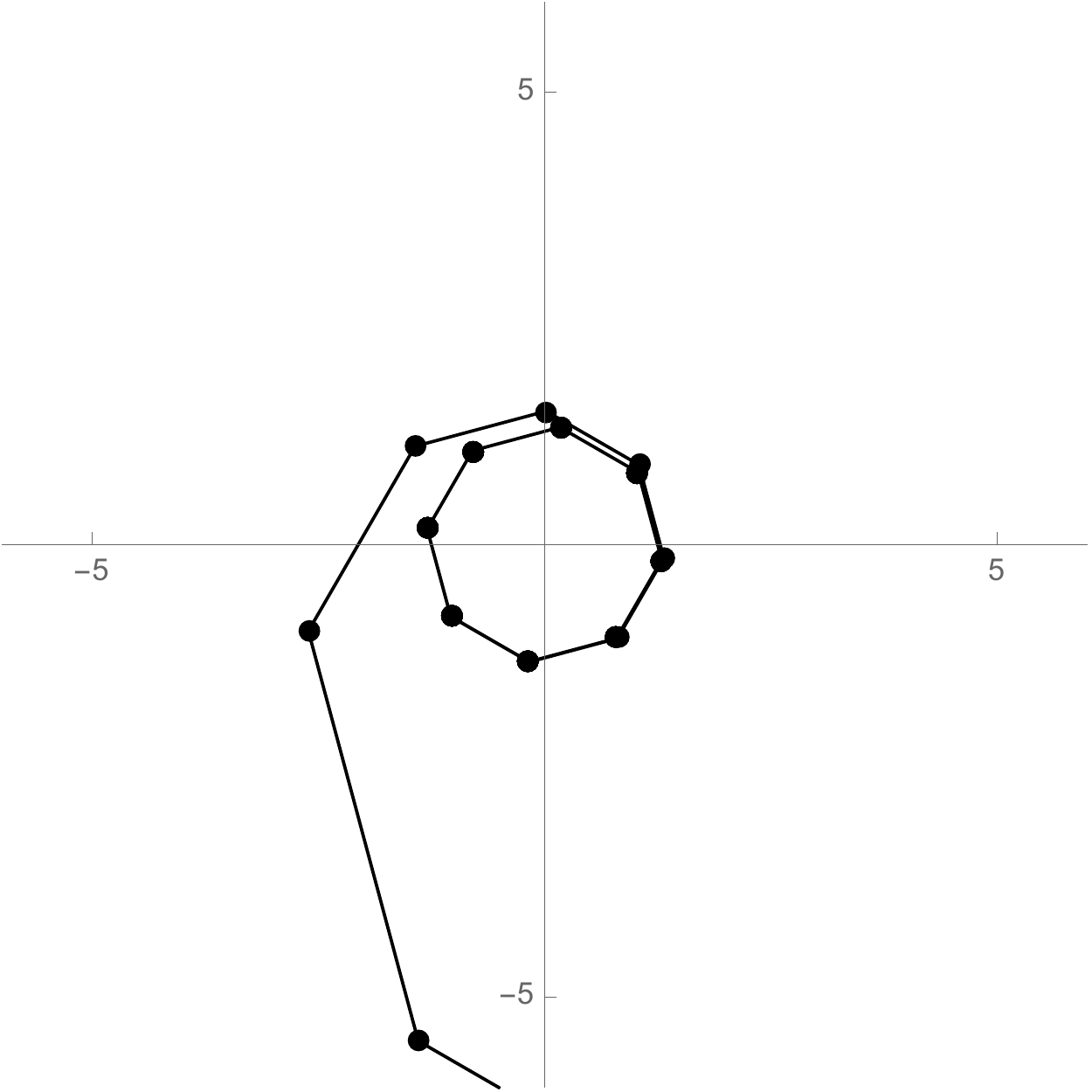}
\end{minipage}
\begin{minipage}{.3\linewidth}
\includegraphics[width=\linewidth]{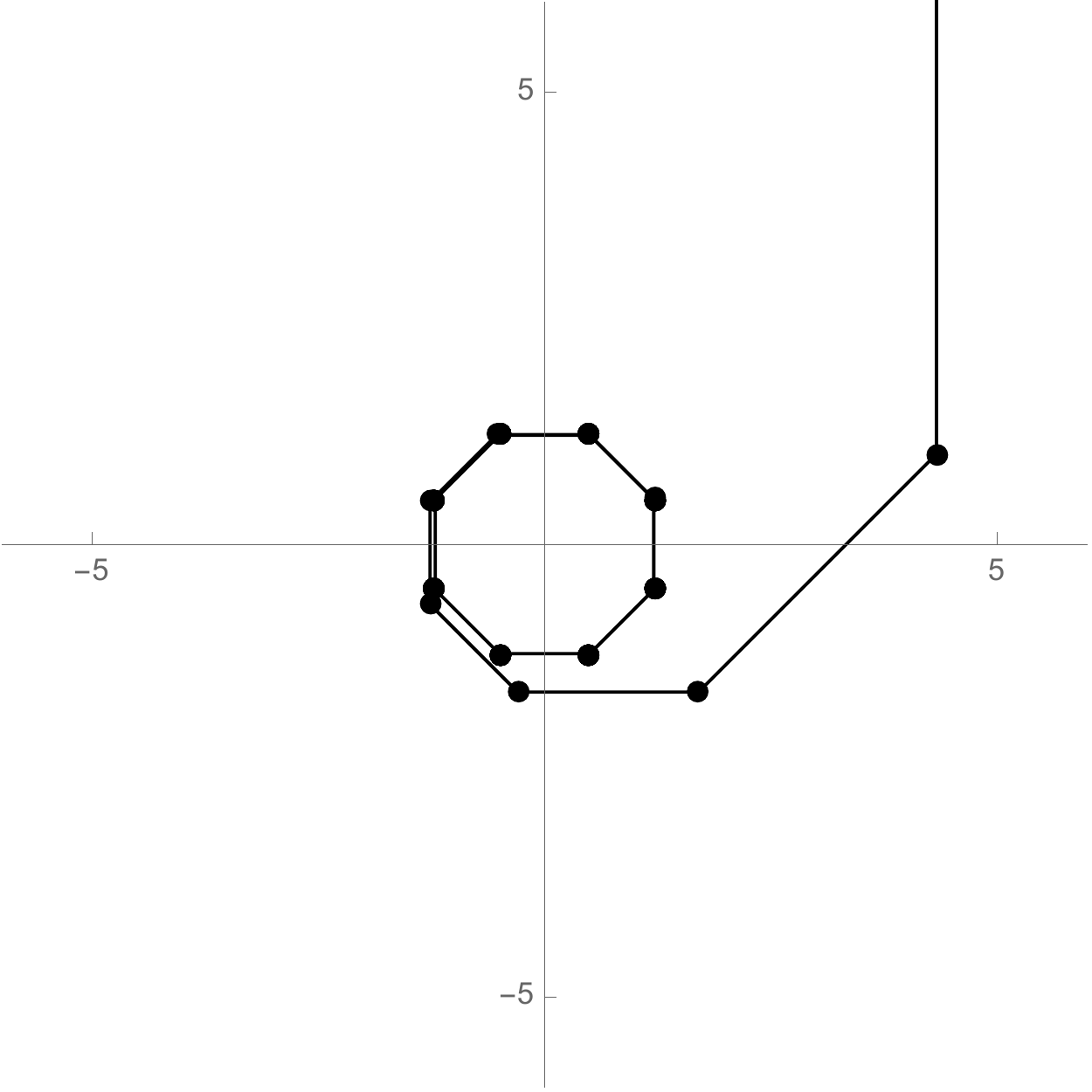}
\end{minipage}
\begin{minipage}{.3\linewidth}
\includegraphics[width=\linewidth]{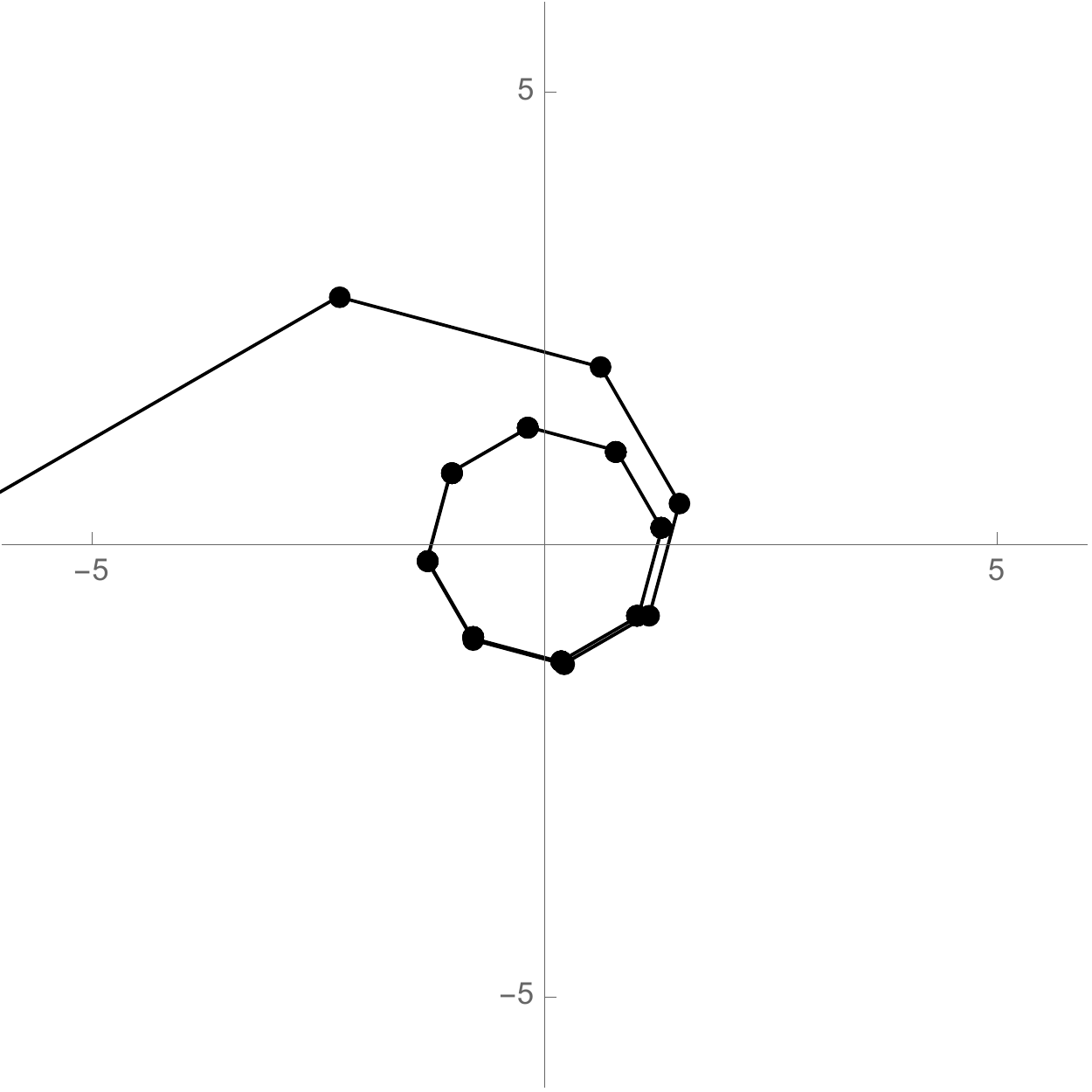}
\end{minipage}
\caption{Motion of discrete plane curves $e^{- \mu_0 m} \gamma^m_n$
corresponding to a 1-shock wave solution of the discrete Burgers
equation \eqref{eqn:auto_dBurgers_with_a}.  Parameters are $\lambda_1 =
1$, and $m=-8$ (left), $0$ (middle), $8$ (right).}
\label{fig:dBurgers_curve-1}
\end{figure}
\end{center}
\begin{center}
\begin{figure}[h]
\begin{minipage}{.3\linewidth}
\includegraphics[width=\linewidth]{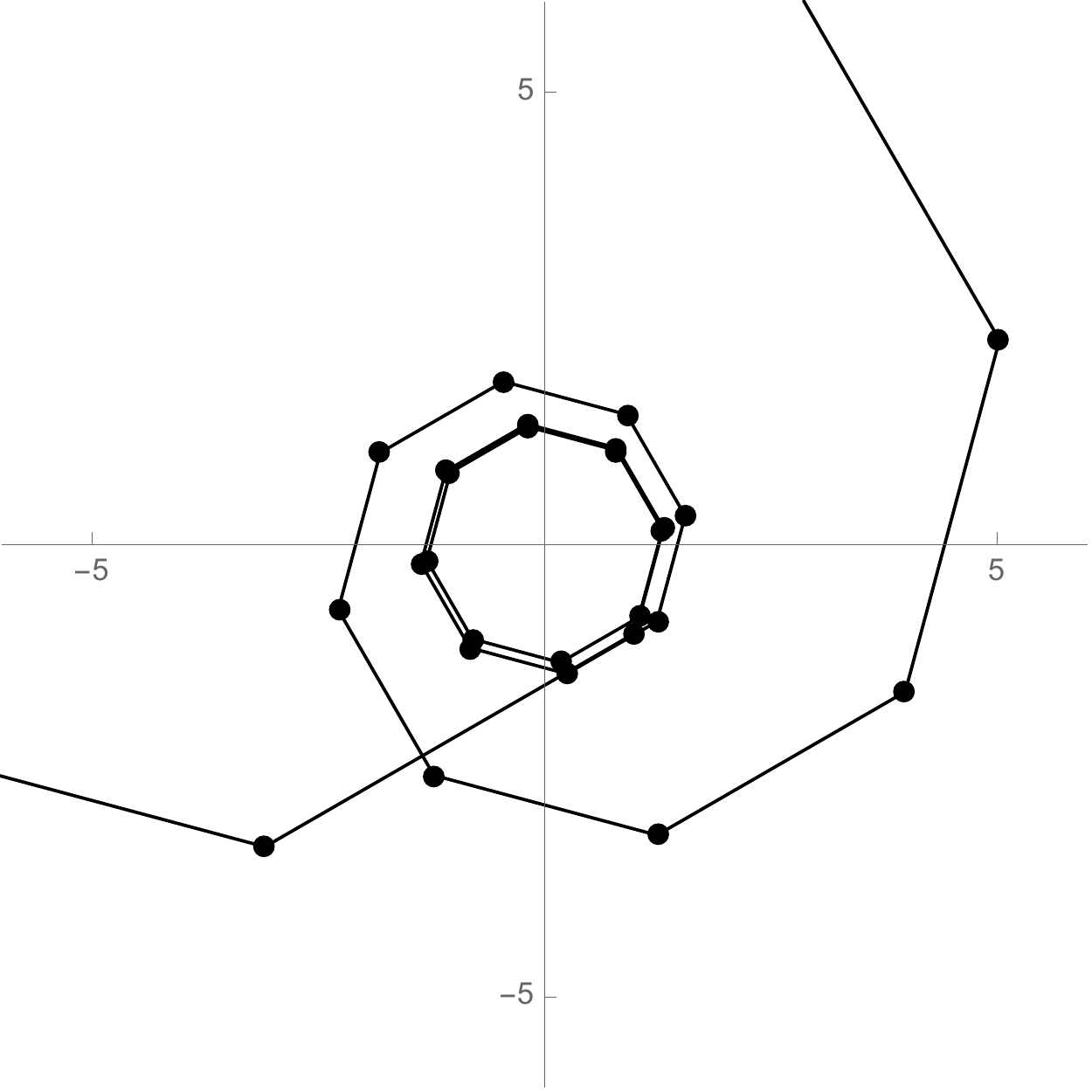}
\end{minipage}
\begin{minipage}{.3\linewidth}
\includegraphics[width=\linewidth]{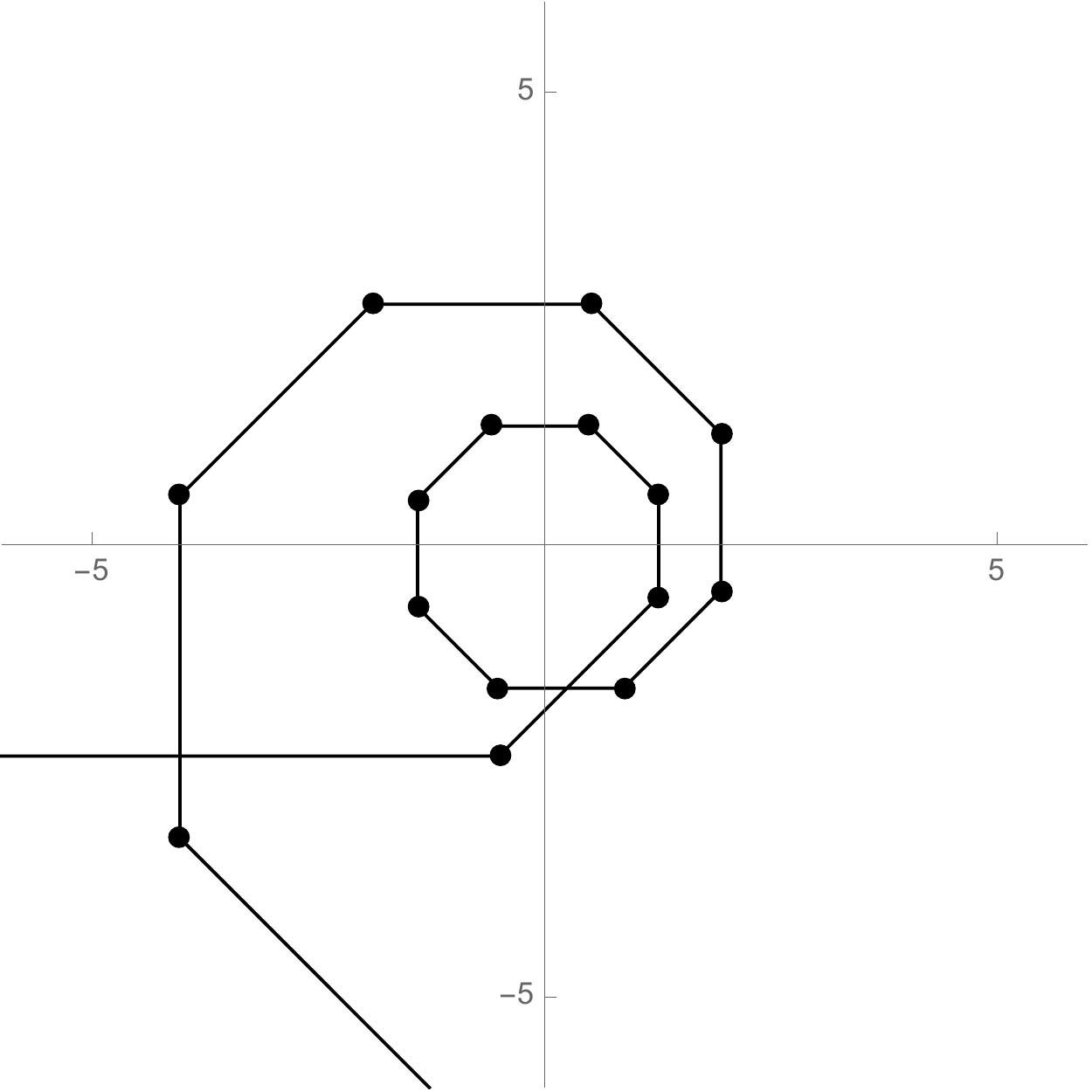}
\end{minipage}
\begin{minipage}{.3\linewidth}
\includegraphics[width=\linewidth]{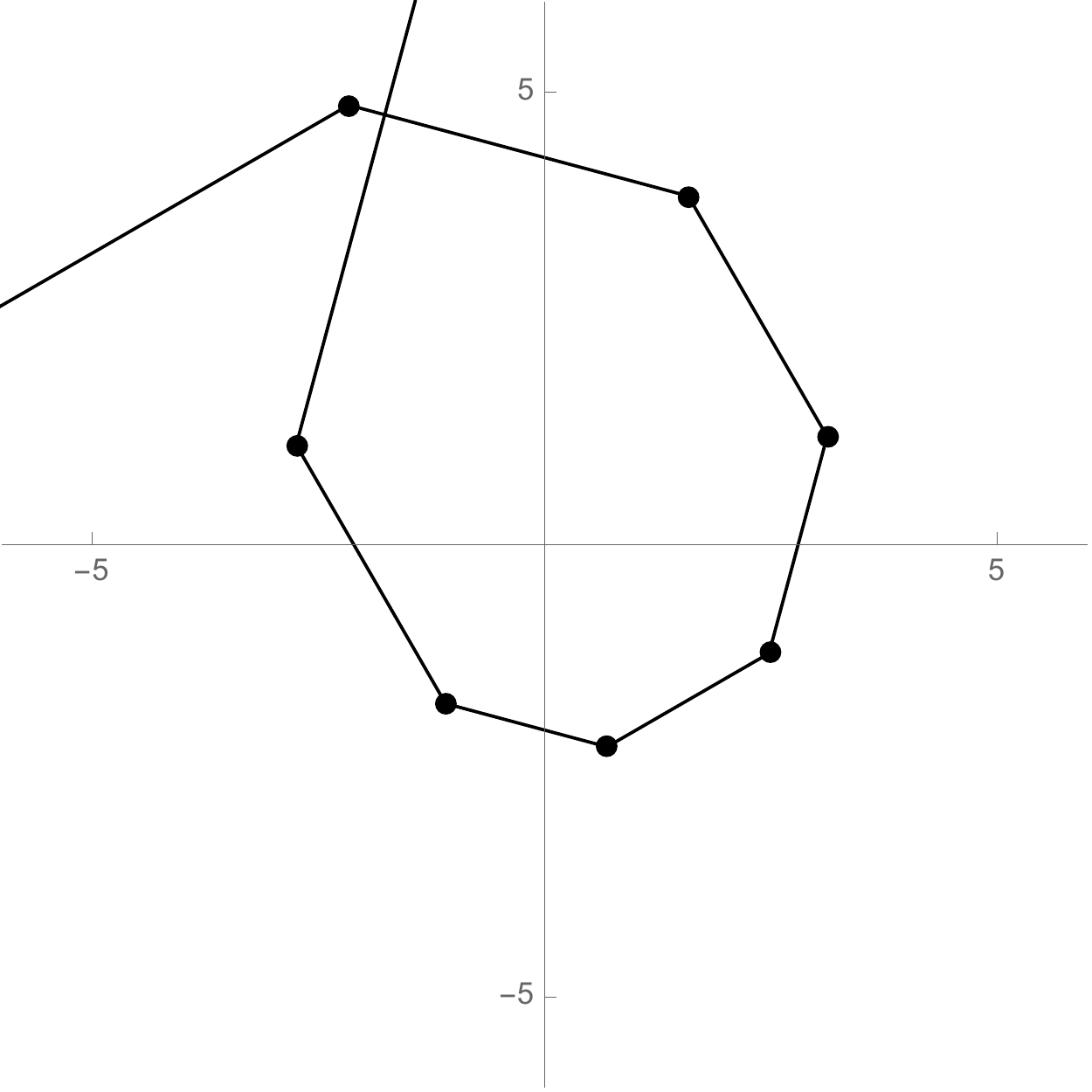}
\end{minipage}
\caption{Motion of discrete plane curves $e^{- \mu_0 m} \gamma^m_n$
corresponding to a 2-shock wave solution of the discrete Burgers
equation \eqref{eqn:auto_dBurgers_with_a}.  Parameters are $\lambda_1 =
1/3$, $\lambda_2 = - 3$, $\xi_2=0$, and $m=-13$ (left), $-6$ (middle),
$-1$ (right).}  \label{fig:dBurgers_curve-2}
\end{figure}
\end{center}

\noindent
\textbf{Acknowledgements.}

The authors would like to thank Professor Jun-ichi Inoguchi for valuable
suggestions. This work was partially supported by JSPS KAKENHI Grant
Number 15K04862.

%%%%%%%%%%%%%%%%%%%%%%%%%%%%%%%%%%%%%%%%%%%%%%%%%%%%%%%%%%%%
%
%
%%%%%%%%%%%%%%%%%%%%%%%%%%%%%%%%%%%%%%%%%%%%%%%%%%%%%%%%%%%%
\appendix

\section{Burgers hierarchy}\label{app:hier}
The Burgers hierarchy is the family of nonlinear partial differential
equations obtained from the linear partial differential equations
\begin{equation}\label{eqn:linear}
\frac{\partial q}{\partial t_i} = \frac{\partial^i q}{\partial
x^i},\quad
i=1,2,3,\ldots
\end{equation}
through the Cole-Hopf transformation
\begin{equation}\label{eqn:Cole-Hopf}
 u = -\frac{\partial }{\partial x}\log q.
\end{equation}
By noticing that
\begin{equation}\label{eqn:Cole-Hopf_inverse} 
q = e^{-\int u~dx},
\end{equation}
the nonlinear equations in the hierachy are expressed \cite{Choodnovsky_Choodnovsky}
as
\begin{equation}\label{eqn:higher_order_Burgers}
\begin{split}
& \frac{\partial u}{\partial t_i}= K_i[u], \\
& K_i[u]= -\frac{\partial }{\partial x}
\left(e^{\int u\,dx} \frac{\partial^i}{\partial x^i}e^{-\int u\,dx}\right).
\end{split}
\end{equation}
Some of the flows of the hierarchy are given by
\begin{align*}
i=1: &\quad  K_1[u]=u',\\
i=2: &\quad  K_2[u]=u''-2uu',\\
i=3: &\quad  K_3[u]=u'''-(u')^2-uu''+3u^2u'.
\end{align*}
An elementary calculation shows the following relation between $K_i[u]$ and $K_{i-1}[u]$:
\begin{equation}
 K_{i}[u] = \Omega K_{i-1}[u], \quad \Omega = \partial_x -u - u'\partial_x^{-1}.
\end{equation}
%In fact, we have
%\begin{align*}
%K_l \left[u\right]
%&=
%- \dfrac{\partial}{\partial x}
%\left(e^{\int u\,dx}
%\frac{\partial^l}{\partial x^l}
%e^{-\int u\,dx}\right)\\
%&=
%- \dfrac{\partial}{\partial x}
%\left(e^{\int u\,dx}
%\dfrac{\partial}{\partial x}
%\frac{\partial^{l - 1}}{\partial x^{l - 1}}
%e^{-\int u\,dx}\right)\\
%&=
%- \dfrac{\partial}{\partial x}
%\left(\dfrac{\partial}{\partial x}
%\left(e^{\int u\,dx}
%\frac{\partial^{l - 1}}{\partial x^{l - 1}}
%e^{-\int u\,dx}\right)
%- \left(\dfrac{\partial}{\partial x} e^{\int u\,dx}\right)
%\frac{\partial^{l - 1}}{\partial x^{l - 1}}
%e^{-\int u\,dx}\right)\\
%&=
%\dfrac{\partial}{\partial x}
%\left(K_{l - 1} \left[u\right]
%+ u\, e^{\int u\,dx}
%\frac{\partial^{l - 1}}{\partial x^{l - 1}}
%e^{-\int u\,dx}\right)\\
%&=
%\dfrac{\partial}{\partial x} K_{l - 1} \left[u\right]
%+ \left(\dfrac{\partial}{\partial x} u\right)
%e^{\int u\,dx}
%\frac{\partial^{l - 1}}{\partial x^{l - 1}}
%e^{-\int u\,dx}
%+ u \dfrac{\partial}{\partial x}
%\left(e^{\int u\,dx}
%\frac{\partial^{l - 1}}{\partial x^{l - 1}}
%e^{-\int u\,dx}\right)\\
%&=
%\dfrac{\partial}{\partial x} K_{l - 1} \left[u\right]
%- \left(\dfrac{\partial}{\partial x} u\right)
%\left(\dfrac{\partial}{\partial x}\right)^{- 1}
%K_{l - 1} \left[u\right]
%- u\, K_{l - 1} \left[u\right]\\
%&=
%\left(\partial_x - u' {\partial_x}^{- 1} - u\right)
%K_{l - 1} \left[u\right]
%\end{align*}
%である.
Here, $\Omega$ is called the recursion operator of the Burgers
hierarchy, by which the equations in the hierarchy can be expressed as
\begin{equation}
  \frac{\partial u}{\partial t_i} = \Omega^{i-1}K_1[u] = \Omega^{i-1}u',\quad i\geq 2.
\end{equation}
%%%%%%%%%%%%%%%%%%%%%%%%%%%%%%%%%%%%%%%%%%%%%%%%%%%%%%%
% Appendix B
%%%%%%%%%%%%%%%%%%%%%%%%%%%%%%%%%%%%%%%%%%%%%%%%%%%%%%%
\section{Discrete Burgers hierachy}\label{sec:discrete_Burgers}
\subsection{Discrete Burgers hierarchy}
Let $\delta, \epsilon$ be constants.
For $i = 0, 1, 2, \ldots$,
we consider the family of linear difference equations
\begin{equation}\label{eqn:linear_autonomous}
\dfrac{q^{m + 1}_n - q^{m}_n}{\delta}
=\widehat{L}^{(i)}[q_n^m],
\end{equation}
where
\begin{equation*}
\widehat{L}^{(i)}[q_n^m]=
\begin{cases}
\Delta^i q^m_n & i=2l,\\
e^{\partial_n/2} \Delta^i q^m_n & i=2l+1.
\end{cases}
\end{equation*}
Here $\Delta$ is a central-difference operator in $n$ defined as
\begin{equation*}
\Delta = \dfrac{e^{\partial_n/2}
- e^{- \partial_n/2}}{\epsilon}.
\end{equation*}
The first few examples of $\widehat{L}^{(i)}[q_n^m]$ are given by
\begin{align*}
i=0: &\quad \widehat{L}^{(0)} \left[q_n^m\right] = q_n^m,\\
i=1: &\quad \widehat{L}^{(1)} \left[q_n^m\right] = \frac{q_{n+1}^m - q_n^m}{\epsilon},\\
i=2: &\quad \widehat{L}^{(2)} \left[q_n^m\right] =\frac{q_{n+1}^m - 2q_n^m + q_{n-1}^m}{\epsilon^2},\\
i=3: &\quad \widehat{L}^{(3)} \left[q_n^m\right] = \frac{q_{n+2}^m - 3q_{n+1}^m + 3q_n^m - q_{n-1}^m}{\epsilon^3}.
%i=4: &\quad \widehat{L}^{(4)} \left[q_n^m\right] = \frac{q_{n+2}^m - 4q_{n+1}^m + 6q_n^m - 4q_{n-1}^m + q_{n-2}^m}{\epsilon^4}.
\end{align*}
The \textit{discrete Burgers hierarchy} is a family of nonlinear
difference equations obtained from \eqref{eqn:linear_autonomous} by the
discrete Cole-Hopf transformation \cite{Fujioka-Kurose:Burgers}
\begin{equation}\label{eqn:discrete_Cole-Hopf}
u^{m}_n = \dfrac{q^{m}_{n + 1}}{q^{m}_n}.
\end{equation}
Noticing that 
\begin{equation*}%\label{eqn:q_by_u}
q^{m}_n = \prod_{k}^{n - 1} u^{m}_k,
\end{equation*}
the $i$-th order equation in the hierarchy can be written as
\begin{equation}\label{d-burgers-higher}
\dfrac{u^{m + 1}_n}{u^{m}_n}
=
\dfrac{1 + \delta \widehat{K}^{(i)} \left[u_{n+1}^m\right]}%
{1 + \delta \widehat{K}^{(i)} \left[u_n^m\right]},
\end{equation}
where
\begin{equation*}%\label{def:Khat}
\widehat{K}^{(i)} \left[u_n^m\right]
=
\frac{1}{q_n^m}
\widehat{L}^{(i)} \left[q_n^m\right].
\end{equation*}
For instance, the first few $\widehat{K}^{(i)}$ are given by
\begin{align*}
i=0: &\quad \widehat{K}^{(0)} \left[u_n^m\right] = 1,\\
i=1: &\quad \widehat{K}^{(1)} \left[u_n^m\right]
=\dfrac{1}{\epsilon}\left(u^{m}_n - 1\right),\\
i=2: &\quad \widehat{K}^{(2)} \left[u_n^m\right]
= \dfrac{1}{\epsilon^2}\left(u^{m}_n - 2 + \dfrac{1}{u^{m}_{n - 1}}\right),\\
i=3: &\quad \widehat{K}^{(3)} \left[u_n^m\right]
= \dfrac{1}{\epsilon^3}\left(u^{m}_{n + 1} u^{m}_n - 3 u^{m}_n
+ 3 - \dfrac{1}{u^{m}_{n - 1}}\right).
\end{align*}
The discrete Burgers hierarchy admits the {\em recursion operators}
which generate higher order flows from lower ones.
\begin{prop}
It holds that 
\begin{equation*}
\widehat{K}^{(i+1)}[u_n^m] = \Omega_n\, \widehat{K}^{(i)}[u_n^m],
\end{equation*}
where $\Omega_n$ is a difference operator defined by
\begin{equation}\label{eqn:recursion_autonomous1}
\Omega_n
=
\begin{cases}
\Omega_n^{(\mathrm{odd})} = \dfrac{1}{\epsilon}
\left(u^m_n e^{\partial_n} - 1\right) & i=2l,\\
\Omega_n^{(\mathrm{even})} = \dfrac{1}{\epsilon}
\left(1 - \dfrac{1}{u^m_{n - 1}} e^{- \partial_n}\right) & i=2l+1.
\end{cases}
\end{equation}
In particular, we have
\begin{equation}\label{eqn:recursion_autonomous3}
\widehat{K}^{(i+2)} \left[u_n^m\right]
=\Omega_n^{(2)} \widehat{K}^{(i)} \left[u_n^m\right],
\end{equation}
where 
\begin{equation}\label{eqn:recursion_autonomous2}
\Omega_n^{(2)}
=
\Omega_n^{(\mathrm{odd})}
\Omega_n^{(\mathrm{even})}
=
\Omega_n^{(\mathrm{even})}
\Omega_n^{(\mathrm{odd})}
=
\dfrac{1}{\epsilon^2}
\left(u^m_n e^{\partial_n}
- 2 + \dfrac{1}{u^m_{n - 1}} e^{- \partial_n}\right).
\end{equation}
\end{prop}
\begin{proof}
Since we have
\begin{equation*}
\widehat{L}^{(2l+1)}[q_n^m] = \frac{e^{\partial_n} - 1}{\epsilon}\,\widehat{L}^{(2l)}[q_n^m],\quad
\widehat{L}^{(2l+2)}[q_n^m]=\frac{1 - e^{-\partial_n}}{\epsilon}\, \widehat{L}^{(2l+1)}[q_n^m],
\end{equation*}
from \eqref{eqn:linear_autonomous},
we obtain
\begin{align*}
q_n^m \widehat{K}^{(2l+1)}[q_n^m]
&= \frac{e^{\partial_n} - 1}{\epsilon} \left(q_n^m \widehat{K}^{(2l)}[u_n^m]\right)
= \frac{1}{\epsilon}\left(q_{n+1}^me^{\partial_n} - q_n^m\right)\widehat{K}^{(2l)}[q_n^m],\\
q_n^m \widehat{K}^{(2l+2)}[u_n^m]
&=\frac{1 - e^{-\partial_n}}{\epsilon} \left(q_n^m\widehat{K}^{(2l+1)}[u_n^m]\right)
=\frac{1}{\epsilon}\left(q_n^m - q_{n-1}^me^{-\partial_n}\right)\widehat{K}^{(2l+1)}[u_n^m],
\end{align*}
which immediately yields \eqref{eqn:recursion_autonomous1}. The second half of the statement can be verified by
a straightforward calculation.
\end{proof}
By using the recursion operators \eqref{eqn:recursion_autonomous1}
and \eqref{eqn:recursion_autonomous2},
$\widehat{K}^{(i)}[u_n^m]$ can be expressed as
\begin{equation*}
\widehat{K}^{(i)} \left[u_n^m\right]
=
\left\{
\begin{array}{ll}\smallskip
\left(\Omega_n^{(2)}\right)^{l}\, 1 & i=2l,\\
\left(\Omega_n^{(2)}\right)^{l}\Omega_n^{({\rm odd})}\, 1 & i=2l+1.
\end{array}
\right.
\end{equation*}
%%%%%%%%%%%%%%%%%%%%%%%%%%%%%%%%%%%%%%%%%%%%%%%%%%%%%%%%
% non-autonomous hierarchy
%%%%%%%%%%%%%%%%%%%%%%%%%%%%%%%%%%%%%%%%%%%%%%%%%%%%%%%%
\subsection{Non-autonomous discrete Burgers hierarchy}
We formulate the discrete Burgers hierarchy with arbitrary lattice
intervals, which we call {\em non-autonomous} discrete Burgers
hierarchy. The hierarchy introduced in the previous
section is sometimes referred to as the {\em autonomous} discrete
Burgers hierarchy. We first introduce the {\em divided difference}
$f[x_j,x_{j+1},\ldots,x_{j+n}]$ of the function $f(x)$ with the base
points $x_j,x_{j+1},\ldots,x_{j+n}$ recursively by
$f[x_j] = f(x_j)$ and
\begin{align*}
\text{first order:}\quad & f[x_j,x_{j+1}] =
\frac{f[x_{j+1}] - f[x_j]}{x_{j+1}-x_j},\\
\text{second order:}\quad & f[x_j,x_{j+1},x_{j+2}] = \frac{f[x_{j+1},x_{j+2}]-f[x_{j},x_{j+1}]}{x_{j+2}-x_j},\\
\text{third order:}\quad & f[x_j,x_{j+1},x_{j+2},x_{j+3}] = \frac{f[x_{j+1},x_{j+2},x_{j+3}]-f[x_{j},x_{j+1},x_{j+2}]}{x_{j+3}-x_j},\\
\text{$n$-th order:}\quad & f[x_j,x_{j+1},\ldots,x_{j+n}] = \frac{f[x_{j+1},\ldots,x_{j+n}]-f[x_{j},\ldots,x_{j+n-1}]}{x_{j+n}-x_j}.
\end{align*}
Among the various properties of the divided differences, we here note the following:
\begin{enumerate}
 \item {\em Expansion formula}. 
%\begin{equation*}
%f[x_j,x_{j+1},\ldots,x_{j+n}] =
%\sum_{k=0}^n \frac{f(x_{j+k})}{\prod\limits_{s=0\atop s\neq k}^{n}(x_{j+k}-x_{j+s})}.%\label{eqn:divided_diff_expansion} 
%\end{equation*}
\begin{equation*}
f[x_j,x_{j+1},\ldots,x_{j+n}] =
\sum_{k=0}^n \frac{f(x_{j+k})}{\prod_{s=0,\, s\neq k}^{n}\,
(x_{j+k}-x_{j+s})}.%\label{eqn:divided_diff_expansion} 
\end{equation*}
For example, we have
\begin{small}
\begin{align*}
f[x_j,x_{j+1}] =\;& \frac{f(x_{j+1})}{x_{j+1}-x_j}
+ \frac{f(x_j)}{x_j-x_{j+1}},\\
f[x_j,x_{j+1},x_{j+2}] =\;& \frac{f(x_{j+2})}{(x_{j+2}-x_{j+1})(x_{j+2}-x_j)} + \frac{f(x_{j+1})}{(x_{j+1}-x_{j+2})(x_{j+1}-x_j)}
+ \frac{f(x_{j})}{(x_{j}-x_{j+2})(x_{j}-x_{j+1})},\\
\begin{split}
 f[x_j,x_{j+1},x_{j+2},x_{j+3}] =\;&
  \frac{f(x_{j+3})}{(x_{j+3}-x_{j+2})(x_{j+3}-x_{j+1})(x_{j+3}-x_{j})}
+ \frac{f(x_{j+2})}{(x_{j+2}-x_{j+3})(x_{j+2}-x_{j+1})(x_{j+2}-x_{j})}\\
& + \frac{f(x_{j+1})}{(x_{j+1}-x_{j+3})(x_{j+1}-x_{j+2})(x_{j+1}-x_{j})}
+ \frac{f(x_{j})}{(x_{j}-x_{j+3})(x_{j}-x_{j+2})(x_{j}-x_{j+1})}.
\end{split}
\end{align*}
\end{small}
As an immediate consequence, it follows that
$f[x_j,x_{j+1},\ldots,x_{j+n}]$ is invariant with respect to
interchanging the base points.

\item {\em Autonomization and continuous limit}.
Putting the lattice interval to be constant, namely,
$x_{j+k}=x_j + k\epsilon$, it follows that
\begin{equation}\label{eqn:divided_diffrence_and_derivative}
f[x_j,x_{j+1},\ldots,x_{j+n}] = \frac{1}{n!}\Delta_{+x}^n\, f(x_j),\quad \Delta_{+x}\,f(x) = \frac{f(x+\epsilon)-f(x)}{\epsilon},
\end{equation}
and thus
\begin{equation*}
f[x_j,x_{j+1},\ldots,x_{j+n}]  \ \longrightarrow\ \frac{1}{n!}\frac{d^nf(x_j)}{dx^n}\quad (\epsilon\to 0).
\end{equation*}
\end{enumerate}
In order to formulate the non-autonomous discrete Burgers hierarchy,
we first introduce the family of
linear difference equations for $q_n^m=q(x_n,t_m)$,\,
$\delta_m=t_{m+1}-t_m$:
\begin{equation}\label{eqn:nonautonomous_linear_eq}
\frac{q_n^{m+1}-q_n^m}{\delta_m} = L_n^{(i)}[q_n^m],
\end{equation}
where
\begin{equation*}%\label{eqn:nonautonomous_linear_eq2}
L_n^{(i)}[q_n^m] =
\begin{cases}
q[x_{n-l},\ldots,x_{n+l}] & i=2l,\\
q[x_{n-l},\ldots,x_{n+l+1}] & i=2l+1.
\end{cases}
\end{equation*}
The first few examples of $L_n^{(i)}[q_n^m]$ are given by
%\begin{equation}\label{eqn:non_autonomous_linear_eq_example}
%\begin{minipage}{0.85\textwidth}
%\begin{footnotesize}
%\begin{displaymath}
\begin{align*}
i=0:\ %\frac{q_n^{m+1}-q_n^m}{\delta_m} 
L_n^{(0)}[q_n^m] =\;& q^m_n,\\
i=1:\ %\frac{q_n^{m+1}-q_n^m}{\delta_m} 
L_n^{(1)}[q_n^m] =\;&
\frac{q_{n+1}^m}{x_{n+1}-x_n} +\frac{q_n^m}{x_{n}-x_{n+1}},\\
i=2:\ %\frac{q_n^{m+1}-q_n^m}{\delta_m}
L_n^{(2)}[q_n^m] =\;&
  \frac{q_{n+1}^m}{(x_{n+1}-x_{n})(x_{n+1}-x_{n-1})} 
+ \frac{q_{n}^m}{(x_{n}-x_{n+1})(x_{n}-x_{n-1})} 
+ \frac{q_{n-1}^m}{(x_{n-1}-x_{n+1})(x_{n-1}-x_{n})} ,\\
i=3:\ %\frac{q_n^{m+1}-q_n^m}{\delta_m}
L_n^{(3)}[q_n^m] =\;&
  \frac{q_{n+2}^m}{(x_{n+2}-x_{n+1})(x_{n+2}-x_{n})(x_{n+2}-x_{n-1})}
+ \frac{q_{n+1}^m}{(x_{n+1}-x_{n+2})(x_{n+1}-x_{n})(x_{n+1}-x_{n-1})}\\
& 
+ \frac{q_{n}^m}{(x_{n}-x_{n+2})(x_{n}-x_{n+1})(x_{n}-x_{n-1})}
+ \frac{q_{n-1}^m}{(x_{n-1}-x_{n+2})(x_{n-1}-x_{n+1})(x_{n-1}-x_{n})}.
\end{align*}
%\end{displaymath} 
%\end{footnotesize}
%\end{minipage}
%\end{equation}
We note that the following recursion relations hold:
\begin{align}
 L_n^{(2l+1)}[q_n^m] &=\label{eqn:nonautonomous_linear_recursion1}
\frac{q[x_{n-l+1},\ldots,x_{n+l+1}] - q[x_{n-l},\ldots,x_{n+l}]}{x_{n+l+1}-x_{n-l}} =\frac{L_{n+1}^{(2l)}[q_{n+1}^m] - L_n^{(2l)}[q_{n}^m]}{x_{n+l+1}-x_{n-l}},\\
L_n^{(2l+2)}[q_n^m] &=\label{eqn:nonautonomous_linear_recursion2}
\frac{q[x_{n-l},\ldots,x_{n+l+1}] - q[x_{n-l-1},\ldots,x_{n+l}]}{x_{n+l+1}-x_{n-l-1}} =\frac{L_n^{(2l+1)}[q_{n}^m] - L_{n-1}^{(2l+1)}[q_{n-1}^m]}{x_{n+l+1}-x_{n-l-1}}.
\end{align}
The non-autonomous discrete Burgers hierarchy is a family of nonlinear
difference equations obtained from \eqref{eqn:nonautonomous_linear_eq}
by the discrete Cole-Hopf transformation \eqref{eqn:discrete_Cole-Hopf}.
The $i$-th order equation in the hierarchy is given as
\begin{equation}\label{eqn:nonautonomous_discrete_Burgers}
 \frac{u_n^{m+1}}{u_n^m} = \frac{1 + \delta_m K_{n+1}^{(i)}[u_{n+1}^m]}{1+ \delta_m K_n^{(i)}[u_{n}^m]},
\end{equation}
where
\begin{equation*}%\label{eqn:nonautonomous_discrete_Burgers_K}
K_n^{(i)}[u_n^m]
= \frac{1}{q^m_n} L_n^{(i)}[q_n^m].
%=\left\{
%\begin{array}{lc}
%{\displaystyle K_n^{(2l)}[u_n^m] = \sum_{j=-l}^{l}\frac{1}{\prod\limits_{s=0\atop s\neq j}^{n}(x_{n+j}-x_{n+s})}\frac{\prod\limits_{k}^{n+j-1}u_{k}^m}{\prod\limits_{k}^{n-1}u_{k}^m}}& i=2l,\\
%{\displaystyle K_n^{(2l+1)}[u_n^m] = \sum_{j=-l}^{l+1}\frac{1}{\prod\limits_{s=0\atop s\neq j}^{n}(x_{n+j}-x_{n+s})}\frac{\prod\limits_{k}^{n+j-1}u_{k}^m}{\prod\limits_{k}^{n-1}u_{k}^m}.} & i=2l+1
%\end{array}
%\right.
\end{equation*}
The recursion operator for the non-autonomous discrete Burgers hierarchy
is given as follows:
\begin{prop}\label{prop:recursion_nonautonomous_Burgers}
It holds that
\begin{equation*}
 K_n^{(i+1)}[u_n^m] = \Omega^{(1,i+1)}_n\, K_n^{(i)}[u_n^m],
\end{equation*}
where $\Omega_n^{(1,i+1)}$ is a difference operator defined by
\begin{equation}\label{eqn:recursion_nonautonomous1}
\Omega_n^{(1,i+1)}
=
\begin{cases}
\dfrac{1}{\epsilon_n^{(i+1)}}\left(u_n^m e^{\partial}-1\right) & i=2l,\\
\dfrac{1}{\epsilon_n^{(i+1)}}\left(1-\dfrac{1}{u_{n-1}^m}
e^{-\partial_n}\right) & i=2l+1,
\end{cases}
\end{equation}
and
\begin{equation}\label{eqn:nonautonomous_epsilon}
\epsilon_n^{(i+1)} =
\begin{cases}
x_{n+l+1}-x_{n-l} & i=2l,\\
x_{n+l+1}-x_{n-l-1} & i=2l+1.
\end{cases}
\end{equation}
In particular, we have 
\begin{equation*}
K_n^{(i+2)} \left[u_n^m\right]
=\Omega_n^{(2,i+2)} K_n^{(i)} \left[u^m_n\right],
\end{equation*}
where 
\begin{align}
\Omega_n^{(2,i+2)}
&=\nonumber
\Omega_n^{(1,i+2)} \Omega_n^{(1,i+1)}\\
&=\label{eqn:recursion_nonautonomous_Burgers}
\begin{cases}
\dfrac{1}{\epsilon_n^{(i+2)}}\left(
\dfrac{u_n^m}{\epsilon_n^{(i+1)}} e^{\partial_n}
- \dfrac{1}{\epsilon_n^{(i+1)}} - \dfrac{1}{\epsilon_{n-1}^{(i+1)}}
+ \dfrac{1}{\epsilon_{n-1}^{(i+1)} u_{n-1}^m} e^{-\partial_n}
\right) & i=2l,\\
\dfrac{1}{\epsilon_n^{(i+2)}}\left(
\dfrac{u_n^m}{\epsilon_{n+1}^{(i+1)}} e^{\partial_n}
- \dfrac{1}{\epsilon_{n+1}^{(i+1)}} - \dfrac{1}{\epsilon_{n}^{(i+1)}}
+\dfrac{1}{\epsilon_{n}^{(i+1)}} \dfrac{1}{u_{n-1}^m} e^{-\partial_n}
\right) & i=2l+1.
\end{cases}
\end{align}
\end{prop}
\begin{proof}
The first half of the statement follows from the recursion relation of
the divided differences.
Indeed, it follows from \eqref{eqn:nonautonomous_linear_recursion1} and
\eqref{eqn:nonautonomous_linear_recursion2} that
\begin{equation*}
L_n^{(2l+1)}[q_n^m] =
\frac{1}{\epsilon_n^{(2l+1)}}\left(e^{\partial_n}-1\right)L_n^{(2l)}[q_n^m],\quad
L_n^{(2l+2)}[q_n^m] =
\frac{1}{\epsilon_n^{(2l+2)}}\left(1-e^{-\partial_n}\right)L_n^{(2l+1)}[q_n^m],
\end{equation*}
which are equivalent to
\begin{align*}
K_n^{(2l+1)}[u_n^m]
 =\frac{1}{\epsilon_n^{(2l+1)}}\left(u_{n}^me^{\partial_n}-1\right)K_n^{(2l)}[u_n^m],\quad
K_n^{(2l+2)}[u_n^m] =\frac{1}{\epsilon_n^{(2l+2)}}\left(1-\frac{1}{u_{n-1}^m}e^{-\partial_n}\right)K_n^{(2l+1)}[u_n^m].
\end{align*}
Thus we have \eqref{eqn:recursion_nonautonomous1}.
The second half is verified by a direct computation.
\end{proof}
\begin{rem}\label{rem:autonomization}\rm
The non-autonomous discrete Burgers hierarchy \eqref{eqn:nonautonomous_discrete_Burgers} reduces to the discrete Burgers hierarchy \eqref{d-burgers-higher} by putting (see \eqref{eqn:divided_diffrence_and_derivative})
\begin{equation*}
x_{n+1}-x_n = \epsilon,\quad L_n^{(i)}[q_n^m] = \frac{1}{i!}\widehat{L}^{(i)}[q_n^m],\quad
K_n^{(i)}[u_n^m] = \frac{1}{i!}\widehat{K}^{(i)}[u_n^m],\quad 
\Omega_n^{(1,i)} = \frac{1}{i}\widehat{\Omega}_n.
\end{equation*}
\end{rem}


\begin{thebibliography}{99}
%%%%%%%%%%%%%%%%%%%%%%%
\bibitem{Bobenko-Suris:book} A.I. Bobenko and Y.B. Suris, Discrete Differential Geometry (American Mathematical Society, Prividence, RI, 2008).
%%%%%%%%%%%%%%%%%%%%%%%
\bibitem{Choodnovsky_Choodnovsky} D.V. Choodnovsky and G.V. Choodnovsky, Pole expansions of nonlinear partial differential equations, Nuovo Cimento {\bf 40B}(1977) 339--353.
%%%%%%%%%%%%%%%%%%%%%%%
\bibitem{Chou-Qu:2002_PD} K.-S. Chou and C.-Z Qu, Integrable equations arising from motions of plane curves, Phys. D {\bf 162} (2002) 9--33.% doi:10.1016/S0167-2789(01)00364-5.
%%%%%%%%%%%%%%%%%%%%%%%
\bibitem{Chou-Qu:2003_JNS} K.-S. Chou and C.-Z Qu, Integrable equations arising from motions of plane curves. II, J. Nonlinear Sci. {\bf 13}(2003) 487--517.% doi:10.1007/s00332-003-0570-0.
%%%%%%%%%%%%%%%%%%%%%%%
\bibitem{Chou-Qu:2004_CSF} K.-S. Chou and C.-Z Qu, Motions of curves in similarity geometries and Burgers-mKdV hierarchies, 
Chaos, Solitons and Fractals {\bf 19}(2003) 47--53.% doi:10.1016/S0960-0779(03)00060-2
%%%%%%%%%%%%%%%%%%%%%%%
\bibitem{Doliwa-Santini:PLA} A. Doliwa and P. M. Santini, An elementary geometric characterization of the integrable motions of a curve, Phys. Lett. {\bf A185}(1994) 373--384.
%%%%%%%%%%%%%%%%%%%%%%%
\bibitem{Doliwa-Santini:JMP} A. Doliwa and P. M. Santini, Integrable dynamics of a discrete curve and the Ablowitz-Ladik hierarchy, J. Math. Phys. {\bf 36}(1995) 1259--1273.
%%%%%%%%%%%%%%%%%%%%%%%
\bibitem{Doliwa-Santini:dsG} A. Doliwa and P. M. Santini, 
Geometry of discrete curves and lattices and integrable difference equations, 
Discrete Integrable Geometry and Physics, A. Bobenko and R. Seiler (eds.), (Clarendon Press, Oxford, 1999) 139--154.
%%%%%%%%%%%%%%%%%%%%%%%
\bibitem{FIKMO:hodograph} B. F. Feng, J. Inoguchi, K. Kajiwara, K. Maruno and Y. Ohta,
Discrete integral systems and hodograph transformations arising from motions of discrete plane curves,
J. Phys. A: Math. Theor. {\bf 44} (2011) 395201(19 pages).
%%%%%%%%%%%%%%%%%%%%%%%
\bibitem{FIKMO:Dym} B. F. Feng, J. Inoguchi, K. Kajiwara, K. Maruno and Y. Ohta, Integrable discretizations of the Dym equation. Front. Math. China {\bf 8} (2013) 1017--1029.
%%%%%%%%%%%%%%%%%%%%%%%
\bibitem{FMO:SP} B. F. Feng, K. Maruno and Y. Ohta, 
Integrable discretizations of the short pulse equation, 
J. Phys. A: Math. Theor.  {\bf 43}(2010), 085203(14 pages).
%%%%%%%%%%%%%%%%%%%%%%%
\bibitem{FMO:CH} B. F. Feng, K. Maruno and Y. Ohta,
A self-adaptive moving mesh method for the Camassa-Holm equation,
J. Comput. Appl. Math. {\bf 235} (2010) 229--243.
%%%%%%%%%%%%%%%%%%%%%%%
\bibitem{FMO:short_CH} B. F. Feng, K. Maruno and Y. Ohta, 
Integrable discretizations for the short-wave model of the Camassa-Holm equation. 
J. Phys. A: Math. Theor. {\bf 43} (2010) 265202(14 pages).
%%%%%%%%%%%%%%%%%%%%%%%
\bibitem{FMO:SP_Lax} B. F. Feng, K. Maruno and Y. Ohta,
Self-adaptive moving mesh schemes for short pulse type equations and their Lax pairs,
Pac. J. Math. Ind. {\bf 6}:8 (2014).
%%%%%%%%%%%%%%%%%%%%%%%
\bibitem{Fujioka-Kurose:Burgers} A. Fujioka and T. Kurose, Motions of curves in the complex hyperbola and the Burgers hierarchy, Osaka J. Math {\bf 45}(2008) 1057--1065.
%%%%%%%%%%%%%%%%%%%%%%%
\bibitem{Hisakado-Nakayama-Wadati} M. Hisakado, K. Nakayama and M. Wadati, Motion of discrete curves in the plane, J. Phys. Soc. Jpn.  {\bf 64}  (1995) 2390--2393. 
%%%%%%%%%%%%%%%%%%%%%%%
\bibitem{Hisakado-Wadati} M. Hisakado and M. Wadati, Moving discrete curve and geometric phase, Phys. Lett. {\bf A214}(1996) 252--258.
%%%%%%%%%%%%%%%%%%%%%%%
\bibitem{Hoffmann:dNLS} T. Hoffmann, Discrete Hashimoto surfaces and a doubly discrete smoke-ring flow, 
Discrete Differential Geometry, A.I. Bobenko, P. Schr\"oder, J.M. Sullivan and	G.M. Ziegler (eds.), Oberwolfach Seminars Vol.39 (Birkh\"auser, Basel, 2008) 95--115.
%%%%%%%%%%%%%%%%%%%%%%%
\bibitem{Hoffmann-Kutz} T.~Hoffmann and N.~Kutz, Discrete curves in $\mathbb{C}P^1$ and the Toda lattice, Stud. Appl. Math.  {\bf 113} (2004) 31--55.
%%%%%%%%%%%%%%%%%%%%%%%
\bibitem{Hopf} Eberhard Hopf, The partial differential equation $u_t+uu_x=\mu u_{xx}$,
Comm. Pure Appl. Math. {\bf 3}(1950) 201--230. 
%%%%%%%%%%%%%%%%%%%%%%%
\bibitem{Inoguchi:MEIS2015} J. Inoguchi, Attractive plane curves in differential geometry, 
MI Lecture Note {\bf 64}(2015) (Kyushu University) 121--124.
%%%%%%%%%%%%%%%%%%%%%%%
\bibitem{IKMO:KJM} J. Inoguchi, K. Kajiwara, N. Matsuura and Y. Ohta, 
Motion and B\"acklund transformations of discrete plane curves, 
Kyushu J. Math. {\bf 66}(2012) 303--324.
%%%%%%%%%%%%%%%%%%%%%%%
\bibitem{IKMO:JPA} J. Inoguchi, K. Kajiwara, N. Matsuura and Y. Ohta, Explicit solutions to the semi-discrete modified KdV equation and motion of discrete plane curves, J. Phys. A: Math. Theor. {\bf 45}(2012) 045206.
%%%%%%%%%%%%%%%%%%%%%%%
\bibitem{IKMO:dmKdV_space_curve} J. Inoguchi, K. Kajiwara, N. Matsuura and Y. Ohta, 
Discrete mKdV and discrete sine-Gordon flows on discrete space curves, 
J. Phys. A: Math. Theor. {\bf 47}(2014) 235202.
%%%%%%%%%%%%%%%%%%%%%%%
\bibitem{Matsuura:IMRN} N. Matsuura, 
Discrete KdV and discrete modified KdV equations arising from motions of planar discrete curves, 
Int. Math. Res. Not. {\bf 2012}(2012) 1681--1698.
%%%%%%%%%%%%%%%%%%%%%%%
\bibitem{Nakayama:JPSJ2007} K. Nakayama, Elementary vortex filament model of the discrete nonlinear Schr\"odinger equation, J. Phys. Soc. Jpn. {\bf 76}(2007) 074003.
%%%%%%%%%%%%%%%%%%%%%%%
\bibitem{Nakayama_Segur_Wadati:PRL} K. Nakayama, H. Segur and M. Wadati, Integrability and the motions of curves, Phys. Rev. Lett. {\bf 69}(1992) 2603--2606.
%%%%%%%%%%%%%%%%%%%%%%%
\bibitem{Nishinari} K. Nishinari, A discrete model of an extensible string in three-dimensional space, J. Appl. Mech.{\bf 66}(1999) 695--701.
%%%%%%%%%%%%%%%%%%%%%%%
\bibitem{Pinkall:dNLS} U. Pinkall, B. Springborn, and S. Wei{\ss}mann, A new doubly discrete analogue of smoke ring flow and the real time simulation of fluid flow, 
J. Phys. A: Math. Theor. {\bf 40} (2007) 12563--12576.
%%%%%%%%%%%%%%%%%%%%%%%
\bibitem{Rogers-Schief:book} C. Rogers and W. K. Schief, 
B\"acklund and Darboux Transformations: Geometry and Modern Applications in Soliton Theory
(Cambridge University Press, Cambridge, 2002).
%%%%%%%%%%%%%%%%%%%%%%%
\bibitem{Shimizu-Sato:JSIAM} M. Sato and Y. Shimizu, Log-aesthetic curves and Riccati equations from the viewpoint of similarity geometry, JSIAM Lett. {\bf 7}(2015) 21--24.
\end{thebibliography}
\end{document}